\documentclass[12pt, a4paper]{article}
\usepackage{amsfonts}
\usepackage{upgreek}
\usepackage{framed}
\usepackage{latexsym, amsmath, amsthm, amssymb,
   color, epsfig,array,accents}
\usepackage{graphicx}
\usepackage[utf8]{inputenc}
\usepackage{csquotes}
\usepackage{siunitx}
\usepackage{extarrows}
\usepackage{datetime}
\usepackage{parskip}
\usepackage{ragged2e}
\usepackage[normalem]{ulem}
\usepackage[colorlinks=true, allcolors=blue]{hyperref}
\usepackage{cleveref}
\crefname{appsec}{Appendix}{Appendices}
\Crefname{appsec}{Appendix}{Appendices}
\usepackage{multirow}
\usepackage{caption}
\usepackage{float}
\newtheorem{theorem}{Theorem}
\newtheorem*{theorem*}{Theorem}
\newtheorem{definition}{Definition}
\newtheorem{lemma}{Lemma}
\newtheorem*{lemma*}{Lemma}
\newtheorem*{proposition*}{Proposition}
\newtheorem{corollary}{Corollary}
\newtheorem*{corollary*}{Corollary}
\newtheorem{proposition}{Proposition}

\newtheorem*{remark*}{Remark}
\usepackage{float}
\usepackage{enumerate}
\usepackage{mathtools}
\usepackage{natbib}
\usepackage[toc,page]{appendix}
\allowdisplaybreaks

\let \cite=\citep

\newdateformat{monthyear}{\monthname[\THEMONTH] \THEYEAR}

\graphicspath{{/Users/georgios.angelis/Desktop/latex/}}

\title{The Matching Function:\\ A Unified Look into the Black Box} 
\author{
Georgios Angelis
\and 
Yann Bramoullé
}
\date{\monthyear\today\thanks{\scriptsize Angelis, georgios.angelis@glasgow.ac.uk: University of Glasgow; Bramoullé, yann.bramoulle@univ-amu.fr: Aix-Marseille University, CNRS, Aix-Marseille School of Economics. We are grateful to Jim Albrecht for his fruitful discussion of the paper at the 2024 Australasian Search and Matching Workshop, and to Matt Elliott, Manolis Galenianos, Sanjeev Goyal, Yannis Ioannides, Leo Kaas, Philip Kircher, Francis Kramarz, Mihai Manea, Sephorah Mangin, Pascal Michaillat, Roland Rathelot, Alireza Tahbaz-Salehi, Andreas Tryphonides, colleagues and seminar participants for very constructive comments and suggestions. Part of this work was conducted while Angelis was based in Aix-Marseille School of Economics and Bocconi University, IGIER; he thanks them both for the excellent work conditions. This work was supported by the French National Research Agency through Grant ANR-17-EURE-0020 and the Excellence Initiative of Aix-Marseille University, A*MIDEX.}}
\begin{document}

\maketitle
\thispagestyle{empty}
\begin{abstract}
\noindent \footnotesize In this paper, we use tools from network theory to trace the properties of the matching function to the structure of granular connections between applicants and vacancies. We unify seemingly disparate parts of the literature by recovering multiple functional forms as special cases including the CES. We derive a testable condition under which matching in any network from the broad class we analyze can be thought ``as if'' it comes from a CES matching function, up to a first-order approximation. We provide a theory of match efficacy in which inequality in search intensities is the key determinant of how well the matching process works. A robust finding of our analysis is that dispersion of search intensities on either side of the market is bad for the matching process. We also show that a rise in the market's mean search intensity can reduce match efficacy when it is associated with a higher Gini coefficient of search intensities.
\vspace{3mm}
\end{abstract}

\newpage\pagenumbering{arabic}
\section{Introduction}
The matching function is the linchpin of most models that depart from the Walrasian equilibrium to capture search frictions in the market \citep{PP2001}. The number of contexts in which it has been used highlights its success: most notably the labor market, where unemployed and vacancies coexist in equilibrium, but also credit markets, goods markets, assets trading over the counter, the new monetarist literature aiming to explain the emergence of money, international trade.

The matching function, however, has remained a ``black box'' for nearly forty years. It is a reduced-form object which economists use for its tractability, yet no systematic analysis has been done of the frictions implicitly assumed to underlie it, such as information limitations and coordination failures. More specifically, little is known about how the structure of these underlying frictions affects the matching function's properties.

This paper proposes a network-based model of the matching process that links the micro to the macro. We consider a broad class of random networks connecting applicants and vacancies, under a benchmark application-and-offer protocol, to unify seemingly disparate parts of the literature and to provide the first systematic study of how the structure of the underlying frictions affects the emergent matching function's functional form and efficacy. 

%

The primitive in our setup is the random network characterized by the collection of applicant-vacancy link-formation probabilities. We take linking to a vacancy to be equivalent to applying to it, and we assume that a job is offered to the most qualified applicant, with qualifications drawn from a common distribution. We remain agnostic about how the network is formed. In this sense, our network is an intermediate, granular object---directly linkable to microdata---whose significance this paper highlights in opening the ``black box.'' 

An applicant may be more likely to apply for a job because it has been advertised to them in a targeted way (e.g. a friend informed them of the position), because they have the relevant skills for it, or because the job is in an area where they are actively searching. The network thus embeds, at the most granular level, precisely the search frictions---informational, geographic, or skill-mismatch---that the literature has implicitly assumed to underlie the matching function.\footnote{Distinct parts of the literature have studied these aspects in isolation: For example \citet{Gale2014} has a separate channel for referrals, with a standard matching function capturing the frictions from all other channels; \citet{MR2018} focus on the effects of geographic restrictions.}

We analyze the implications of heterogeneous search intensities both on the applicant side and on the vacancy side, as well as when vacancies are aggregated into ``locations,'' thereby reducing coordination frictions in the market. Search intensity refers to the mean number of links an applicant or vacancy has to the other side. For our main results, we let search intensities be drawn from a common distribution, so they only vary ex post. Changes in the structure of search frictions correspond to changes in the properties of this underlying distribution. These may involve changes in the type of the distribution, such as moving from the degenerate distribution to a Pareto, or changes within a given type, such as first- or second-order stochastic dominance shifts.  

Our first contribution is to recover multiple meeting technologies as special cases. We unify, in a single framework, the matching function employed by \citet{KW1993}; the urn-ball matching function \citep[e.g.,][]{JKK2000, BSW2001}; the matching function of \citet{Lagos2000}; the matching process of \citet{Shimer2007} and \citet{Mortensen2009}; the matching function of \citet{Kaas2010} and the process underlying \citet{BSWee2024}; the class of ``invariant'' meeting technologies \citep[e.g.,][]{CGW2025}, as well as the constant elasticity of substitution (CES) functional form \citep[e.g.,][]{dHRW2000}. In all cases covered by our analysis, the matching probability\footnote{On the applicant side we compute the job-finding probability; on the vacancy side we compute the vacancy-filling probability. ``Matching probability'' refers to either of them without specifying the side.} is given by a generally applicable, closed-form expression.

Regarding the CES functional form, perhaps the most tractable specification for discrete-time macro models, we also provide a condition under which matching in any network from the broad class we analyze can be thought ``as if'' it comes from such a matching function, up to a first-order approximation. This result echoes \citet{Steve2007}, who, using a specific queuing-system primitive, derives the CES form when the system's parameters scale in a particular way. Our result delivers the CES as a first-order approximation for a whole family of networks, rather than relying on a specific network, and our network of connections is observable in application data, making the condition directly testable.

Our second contribution is to provide a theory of match efficacy, the residual term of the matching function that governs how well the matching process works. The matching probability is a key driver of unemployment over the business cycle \citep[e.g.,][]{Pissa1986, Shimer2012}. Changes in match efficacy contribute to changes in the matching probability over and above market tightness, yet its theoretical foundations remain limited. Our analysis uncovers the effects of inequality in search intensities on match efficacy, providing sharper theoretical underpinnings to existing empirical findings and making novel, testable predictions. 

We show that dispersion, in the sense of second-order stochastic dominance shifts in the distribution of search intensities, is unambiguously bad for matching. This is a very robust finding of our analysis: it holds for search heterogeneity on both the applicant and the vacancy sides, as well as when vacancies are aggregated into ``locations,'' in which case the shifts concern the distribution of market tightness across locations. 

The result on the dispersion of applicants' search intensities is novel, as the standard matching function does not allow for effects of search intensity beyond the mean \citep[e.g.,][ch. 5]{Pissa2000}, and more micro-based approaches \citep[e.g.,][]{C-AZ2005, AGV2006} focus on setups that preclude dispersion. This result complements findings on the compositional effects of the applicant pool on match efficacy \citep{BF2015, HS-W2018}, which relate to the mean.

The result on the dispersion in market tightness across locations generalizes the ``dispersion effect'' of \citet{BF2015}, which is obtained taking a second-order Taylor approximation. The result on the dispersion in vacancy-advertising intensity relates to the findings on the effects of recruiting intensity on match efficacy \citep[e.g.,][]{DFH2013, GSV2018, C-TGK2023}. Contrary to previous contributions, our results do not require specifying a particular matching function or an exogenous efficacy term, at any level of aggregation.

We also show that increases in the market's mean search intensity, corresponding to first-order stochastic dominance shifts in the distribution of search intensities, generally have an ambiguous effect on matching. Among the cases we analyze, the mean job-finding probability exhibits an inverted-U shape as a function of mean search intensity when associated with a rise in the Gini coefficient of the distribution; in all other cases it is increasing. 

The inverted-U shape for the job-finding probability has appeared in the social networks literature \citep[][]{C-A2004} and is featured in the setup of \citet{AGV2006}. Its existence beyond these fully symmetric cases, and its connection to the Gini coefficient, are novel. Such an ambiguous effect is entirely absent from macro matching functions, which are monotonically increasing in mean search intensity \citep[e.g.,][ch. 5]{EMR2015, MPS2018, Pissa2000}.

As our interest in this paper is in the aggregate matching function, we mostly focus on the large market, where applicants and vacancies become negligible in size, and the sizes of the two sides of the market go to infinity. As we illustrate, however, our approach is applicable to both small markets \citep[e.g.,][]{BSW2001} and large ones \citep[e.g.,][]{Shimer2007}. It also applies to both sparse networks, where mean search intensity is finite (presumably the case for most applications), and dense networks---the large-market generalization of the urn-ball setup---where mean search intensity becomes infinite \citep[e.g.,][]{Sephorah2025}.

The rest of the paper is structured as follows. \Cref{sec:setup} introduces the environment and necessary terminology. \Cref{sec:hetero-apps,sec:implic} are the main sections of the paper: the former derives our generally applicable, closed-form expression for the job-finding probability under applicant-side search heterogeneity; the latter derives most of the special cases and analyzes the effects of inequality in search intensities. \Cref{sec:locations} examines the case in which vacancies are aggregated to ``locations,'' and \Cref{sec:hetero-jobs} analyzes search heterogeneity on the vacancy side. \Cref{sec:conclusion} offers a discussion and concludes.

\section{The setup}\label{sec:setup} 
We start by introducing the economic environment and some necessary terminology. The primitive of the environment is a (bipartite) \textit{random network}, characterized by two sets $\mathcal{U}, \mathcal{V}$ and a matrix $\big(p_{ij}\big)$ denoting the probability of any link $ij$ being formed, where $i \in \mathcal{U}, \ j\in \mathcal{V}$. As a convention, we will index the elements of $\mathcal{U}$ by $i=1,2, ..., U$ and the elements of $\mathcal{V}$ by $j=1,2, ...,  V$, where $U, V \in \mathbb{N}$ are the sizes of the two sides of the market.

We take the elements of $\mathcal{U}$ to correspond to applicants, i.e. workers searching for a job, and the elements of $\mathcal{V}$ to correspond to job postings by firms. With slight abuse of terminology we will refer to $\mathcal{U}$ as containing the \textit{unemployed}, and $\mathcal{V}$ as containing the \textit{vacancies}.\footnote{Strictly speaking, an ``applicant'' may not be classified as ``unemployed'' in the data, for example when they search on the job. Similarly the same vacancy may have multiple job postings corresponding to it. Since the matching function in its most common form is specified between the ``unemployed'' and ``vacancies'' \citep[e.g.,][ch. 1]{Pissa2000}, we adopt this terminology.} 

Links $ij$ are independent Bernoulli random variables, where
\[
Y_{ij} = 
\begin{cases}
1, &\text{w.p.} \ p_{ij}\\
0, &\text{w.p.} \ 1-p_{ij}
\end{cases}
\]

We will also need the two following definitions of network objects.

\begin{definition}[Degrees] Applicant $i$'s degree, defined as $\sum_j Y_{ij}$, is the number of vacancies the applicant connects to. Similarly vacancy $j$'s degree, defined by $\sum_i Y_{ij}$, corresponds to the number of applicants the vacancy connects to. 
\end{definition}

\begin{definition}[Indirect degrees] Conditional on linking to a vacancy $j$, an applicant $i$'s indirect degree at that vacancy, defined by $\sum_{k \neq i} Y_{kj}\ | \ Y_{ij}=1$, is the number of other (competing) applicants linking to the same vacancy. 
\end{definition}

Naturally, as sums of random variables, degrees and indirect degrees are also random variables and as we will illustrate in our analysis they follow known distributions. We will generally denote degrees by $d$ and indirect degrees by $\tilde{d}$ appropriately indexed. 

\textbf{Example of a network}. Suppose there are $2$ applicants and $2$ job postings in the market, and all links are formed with the same probability $p$.\footnote{This is a 2-by-2 example of the (bipartite) Erdős-Rényi network, a benchmark in random networks.} Figure 1 depicts one realization of that network, which occurs with probability $p^3(1-p)$.
\begin{figure}[H]
  \centering
    \includegraphics[width=0.5\textwidth]{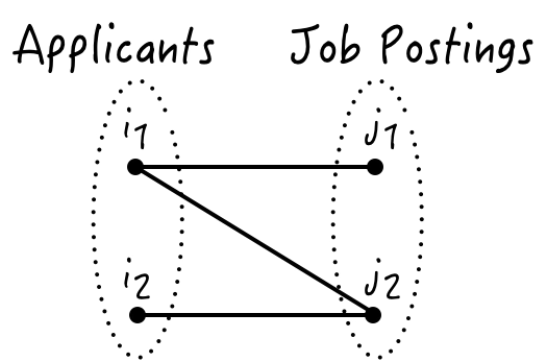}
    \caption{A 2-by-2 example.}
\end{figure}
The degrees of both applicants and both jobs all follow the same binomial distribution $Bin(2,p)$, of two trials each with probability of success $p$; the realized degree of applicant $i_1$ is $2$ and that of job $j_1$ is $1$; those of $i_2, j_2$ are $1$ and $2$ respectively. The indirect degree of applicant $i_1$ at both vacancies is a Bernoulli random variable with success probability $p$; the realized indirect degree of applicant $i_1$ at vacancy $j_1$ is $0$ and at vacancy $j_2$ it is $1$.

As in what follows we will be making connections to the search and matching literature, let us also define a central quantity of that literature, \textit{market tightness}, $\theta = \frac{V}{U}$.

\textbf{Applications, offers and interpretation}. We assume $p_{ij}$ captures the probability applicant $i$ applies for job $j$. As this is a primitive in our setup we remain agnostic as to how $p_{ij}$'s are formed: applicant $i$ may be more likely to apply for job $j$ because a friend informed them of the position, or because they have the skills to do that type of job, or because job $j$ is located in a geographical area where applicant $i$ is actively looking for jobs, or because job $j$ has been advertised in a targeted way to applicants of applicant $i$'s profile. 

The network, characterized by the collection of the applicant-vacancy link-formation probabilities $\big(p_{ij}\big)$, can thus be taken to embed at the most granular level precisely the search frictions---be they informational, geographic or skill-mismatch frictions---which the search and matching literature has been assuming to implicitly underlie the matching function. The network's \textit{structure}, captured in the properties of the $\big(p_{ij}\big)$ matrix, will allow us to induce a notion of structure to the underlying search frictions, and analyze its implications.

Throughout our analysis, we will assume a job will be offered uniformly at random to one of the applicants that have applied to it. As the following remark shows, this can be thought of as the outcome of a process where the job is offered to the best qualified applicant, assuming qualities are drawn independently at each vacancy from an underlying distribution.

\begin{remark*} Suppose a vacancy $j$ receives $d_j$ applications, and each is associated with a measure of quality, drawn independently from some underlying distribution. Then, the probability that any given applicant is the most qualified for the position is $\frac{1}{d_j}$.
\end{remark*}
\begin{proof} See \Cref{sec:app-B}.
\end{proof}
We will assume a vacancy is offered to the most qualified candidate and to them only. That is we will not consider cases where firms go down their applicant list to offer the vacancy to less qualified candidates if their top candidate takes another job, as we will not consider the existence of some absolute quality level that acts as a cutoff for application screening. We return to these modeling choices within a broader discussion in \Cref{sec:conclusion}. 


\section{Matching with heterogeneous search intensity}\label{sec:hetero-apps}
We start by analyzing the effects of heterogeneous search effort across applicants. To this end, we will assume that link-formation probabilities are applicant-specific, shutting down heterogeneity on the vacancy side. Denote
\[
p_{ij} = p_i
\]
so the matrix $\big(p_{ij}\big)$ collapses to a vector $\vec{p}=(p_1, p_2,..,p_U)$. Note that $p_i$ can be interpreted as a continuous measure of applicant $i$'s \textit{search intensity}: the higher the $p_i$, the higher the probability of applicant $i$ forming links, and the higher the expected number of links for $i$.

This specification treats all vacancies as ex ante homogeneous, while it allows for arbitrary search heterogeneity on the applicant side. This will be our benchmark specification, until \Cref{sec:hetero-jobs}, where we will analyze linking heterogeneity on the vacancy side. 

A series of intermediate results will lead to the main theorem of this section which gives a very general and compact expression for the matching function.

\begin{lemma}\label{link-proba} When linking probabilities are given by $p_{ij}=p_i$, applicant degrees $d_i \equiv \sum_j Y_{ij}$ are independent and each follows a binomial distribution, $d_i \sim Bin(V, p_i)$. Vacancy degrees $d_j \equiv \sum_i Y_{ij}$ are i.i.d. following a Poisson binomial distribution, $d_j \sim PB(\vec{p})$, where $\vec{p}=(p_1, p_2,...,p_U)$. 
\end{lemma}
\begin{proof}
See \Cref{sec:app-C}.
\end{proof}

As it is one of the less common distributions, let us note that the ``Poisson binomial'' distribution is a generalization of the binomial distribution, where a random variable is a sum of independent Bernoulli trials with (potentially) different probabilities of success; we give a self-contained presentation of it, with its relevant properties in \Cref{sec:app-A}.

Given that applicant degrees are binomially distributed, mean applicant degrees denoted by $\bar{d}_i$, are given by $\bar{d}_i = p_i V$. Mean degrees are the same across vacancies; we will denote this common value by $\bar{d}_V$ and it is the mean of the Poisson binomial, $\bar{d}_V = \sum_i p_i$.

\begin{corollary} When linking probabilities are given by $p_{ij}=p_i$, applicant $i$'s indirect degrees at any vacancy are i.i.d. following the Poisson binomial distribution $PB(\vec{p}_{-i})$, where $\vec{p}_{-i} = (p_1, p_2,...,p_{i-1},p_{i+1},...,p_U)$. We will denote $\tilde{d}_{i} \equiv \sum_{k \neq i} Y_{kj}$, $\forall j$.
\end{corollary}
\begin{proof}
See \Cref{sec:app-C}.
\end{proof}
We can now compute the job-finding probability of any applicant $i$.

\begin{lemma}\label{f_i-small} When linking probabilities are given by $p_{ij}=p_i$, applicant $i$'s job-finding probability is given by
\[
f_i = 1 - \big(1-p_i\phi\big(\vec{p}_{-i}\big)\big)^V
\]
where $\phi\big(\vec{p}_{-i}\big) \equiv \mathbb{E}\left[\frac{1}{1+\tilde{d}_i}\right]$, and $\tilde{d}_i \sim PB(\vec{p}_{-i})$. 
\end{lemma}
\begin{proof}
See \Cref{sec:app-C}.
\end{proof}

For intuition, let us note that the expression $f_i = 1 - \big(1-p_i\phi(\vec{p}_{-i})\big)^V$ is interpretable: the probability to find a job is $1$ minus the probability of receiving no offers. Inside the big parenthesis is the probability to not receive an offer from one (any one) of the $V$ vacancies; these are i.i.d. events, hence their joint probability is raised to the power of $V$. The probability of each of these events is $1$ minus the intersection of two independent events: to link to that vacancy, which happens with probability $p_i$, and to receive an offer from it, which happens with probability $\phi(\vec{p}_{-i})$.

\Cref{f_i-small} has as special case the functional form of the urn-ball setup, the first and perhaps most classic microfoundation of the matching function \citep[e.g.,][]{But1977, Hall1979, JKK2000, BSW2001, PP2001}.\footnote{This is the setup of \citet{BSW2001}, when (their) buyers are mapped to (our) vacancies; note the equivalence of ``the mixed strategy each buyer uses must be to visit \textit{all sellers} with equal probability'' \citep[][p. 1067 italics our own]{BSW2001}, with $p_i=1, \forall i$ and a vacancy being offered uniformly at random to an applicant. If one makes the opposite mapping, the vacancy-filling probability has the same functional form with $U$ and $V$ swapped \citep[e.g.,][]{PP2001}; this case amounts to each applicant sending a single application, and we are not considering it in this paper.}

\begin{corollary} In the special case when linking probabilities are given by $p_i = 1, \forall i$, the job-finding probability is given by
\[
f = 1 - \left(1-\frac{1}{U}\right)^V
\]
\end{corollary}
\begin{proof}
See \Cref{sec:app-C}.
\end{proof}

Going forward we will also need the following result.\footnote{This is really an accounting identity that can be shown to hold for any bipartite network.}
\begin{lemma}\label{lemma-acc-id} The market-level average of mean applicant degrees $\bar{d}_U \equiv \frac{\sum_i \bar{d}_i}{U}$, the vacancies' mean degree $\bar{d}_V$, and the market tightness $\theta = \frac{V}{U}$, satisfy $\bar{d}_U = \theta \bar{d}_V$.
\end{lemma}
\begin{proof}
See \Cref{sec:app-C}.
\end{proof}

Since our object of interest is the aggregate matching function, going forward our results will focus on \textit{the large market} limit. Roughly speaking, this is when the size of the market grows large, and each applicant and vacancy become negligibly small. Formally:

\begin{definition}[Large market] The large market is the limit of $U,V \to \infty$, while linking probabilities $p_i \to 0$. The limit value of market tightness $\frac{V}{U}$ is assumed to exist and be equal to $\theta$; the limit value of mean degree $p_i V$ is assumed to exist and be equal to $\bar{d}_i$, $\forall i$; the limit value of $\frac{\sum_i \bar{d}_i}{U}$ is assumed to exist and be equal to $\bar{d}_U$.
\end{definition}

Here is the analog of \Cref{link-proba} in the large market limit.

\begin{lemma} In a large market, applicant degrees are independent each following a Poisson distribution with mean $\bar{d}_i$. Vacancy degrees are i.i.d. following a Poisson distribution with mean $\bar{d}_V$.
\end{lemma}
\begin{proof}
See \Cref{sec:app-C}.
\end{proof}

Notice that in the large market, where the binomials have converged to Poissons, the primitive is no longer the vector of $\vec{p}$, but rather the collection of $\{\bar{d}_i\}_{i=1}^{\infty}$. Naturally, $\bar{d}_i$ also captures applicant $i$'s search intensity, and we will be referring to it as such. 

The following result is the analog of \Cref{f_i-small} in the large market limit.

\begin{lemma} In a large market, when search intensities are given by $\{\bar{d}_i\}_{i=1}^{\infty}$, applicant $i$'s job-finding probability is given by
\[
f_i = 1-e^{-\bar{d}_i\phi}
\]
where $\phi \equiv \frac{1-e^{-\bar{d}_V}}{\bar{d}_V}$. 
\end{lemma}
\begin{proof}
See \Cref{sec:app-C}.
\end{proof}

We can now prove our first main result for the matching function. So far, we have allowed for rich heterogeneity across applicants, allowing them to differ ex ante in their search intensities. To get closer to the standard matching functions which do not distinguish among applicants, we will let $\{\bar{d}_i\}$ be drawn i.i.d. from a distribution, thus making applicants ex ante identical, allowing only for ex post heterogeneity. Yet, as we will show, this setup retains enough structural richness for our analysis to go beyond the special cases analyzed in the literature.

\begin{theorem}\label{thm-1} In a large market, when applicant search intensities $\{\bar{d}_i\}_{i=1}^\infty$ are drawn i.i.d. from a distribution $G$ with a finite mean $\bar{d}_U$, support in a subset of $[0, \infty)$, and a moment-generating function (mgf) well-defined for non-positive values, the mean job-finding probability is given by
\[
f(\theta; G) = 1- M_{\bar{d}}(-\phi), \tag{1}
\]
where $M_{\bar{d}}(t) \equiv \mathbb{E}[e^{\bar{d} t}]$ is the moment-generating function (mgf) of $\bar{d}_i, \ \forall i$, and $\phi \equiv \frac{1-e^{-\bar{d}_U/\theta}}{\bar{d}_U/\theta}$. 
\end{theorem}
\begin{proof}
See \Cref{sec:app-C}.
\end{proof}

Given that $0 \leq e^{-\bar{d}\phi} \leq 1$ for $\bar{d}\phi \geq 0$, it can readily be confirmed that $f(\theta; G)$ takes values in $[0,1]$. It can also be shown that it satisfies the ``standard'' properties of reduced-form matching technologies with respect to $\theta$ \citep[e.g.,][]{C-TGK2023}:

\begin{proposition} The function $f(\theta; G)$ of \Cref{thm-1} is increasing and concave in $\theta$ and satisfies $f(0; G)=0$.
\end{proposition} 
\begin{proof}
See \Cref{sec:app-C}.
\end{proof}

Expression (1) is a very general and compact formula that gives the expected job-finding probability for a broad family of networks, namely any distribution $G$ satisfying relatively mild conditions, from which $\{\bar{d}_i\}_{i=1}^\infty$ are drawn. In what follows, we will apply it to the cases of the degenerate, exponential, gamma, Pareto and uniform families of distributions.

In the environment of \Cref{thm-1}, applicant degrees are by construction i.i.d., following 
\[
P(d = k) = \int_0^\infty \frac{\bar{d}^k}{k!}e^{-\bar{d}}dG(\bar{d})
\]
Such a distribution is known as a \textit{mixed Poisson} distribution, which is a special case of a mixture of distributions. The mixed Poisson has recently found applications on the intersection of networks and IO \citep{CUZ2024}, extreme value outcomes \citep{Sephorah2025}, as well as the labor market \citep{CGW2025}.\footnote{Mixtures of distributions, generally $P(d) = \int P_F(d|\bar{d})dG(\bar{d})$, are featured for example in Bayesian hierarchical models \citep[e.g.,][]{Lancaster2004}. The mixed Poisson is the special case where $F$ is the Poisson distribution. The cited papers give the pmf of the mixed Poisson distribution in an equivalent form, which in our notation reads $P(d = k) = \int_0^\infty \frac{(\bar{d}_U\bar{d}_n)^k}{k!}e^{-\bar{d}_U\bar{d}_n}dG(\bar{d}_n)$, where $\bar{d}_n$ is $\bar{d}$ normalized to have a mean of $1$; to go from one form to the other one needs to apply the change of variable formula for pdfs, with $\bar{d}=\bar{d}_U\bar{d}_n$.} The last two refer to two-sided markets and are directly comparable with our setup. We show that our setup embeds the search processes therein as a special case. 

\section{Implications of heterogeneous search intensities}\label{sec:implic}
We start by analyzing special cases of interest. We then give comparative statics of how changes in the distribution of search intensities $G$ affect the efficacy of the matching process.

\subsection{The CES functional form and other special cases}\label{subsec:spec-cases}
\textbf{The degenerate distribution}. Suppose $G$ is degenerate at $\bar{d}_U$. Its mgf is $e^{\bar{d}_U t}$, hence (1) becomes
\[
f = 1-e^{-\bar{d}_U\phi}
\]
This is the matching function in \citet{Kaas2010}. In this case all applicant degrees are drawn from the same Poisson distribution with mean $\bar{d}_U$ and, as always in our benchmark specification, all vacancy degrees are drawn from the same Poisson distribution with mean $\bar{d}_V$. This is also the matching process used by \citet{BSWee2024}.

A degenerate $G$ gives the large market version of the \textit{Erdős-Rényi} network we gave in our 2-by-2 example in \Cref{sec:setup}, an important benchmark among random networks. It also holds special significance among the family of networks covered by \Cref{thm-1}: for the purposes of matching it can be thought as a 1st-order approximation to any such network.
\begin{proposition}\label{prop-1st-ord} Taking a 1st-order Taylor expansion with respect to $\bar{d}$ around $\bar{d}_U$, the job-finding probability of \Cref{thm-1} becomes
\[
f \stackrel{\text{\tiny 1st}}{\approx} 1-e^{-\bar{d}_U\phi}
\]
\end{proposition}
\begin{proof}
See \Cref{sec:app-D}.
\end{proof}

\textbf{``Dense'' vs ``sparse'' networks}. We defined the large market for the case of \textit{sparse networks}, when $\bar{d}_U < \infty$. We can relax this assumption, letting $\bar{d}_U \to \infty$ to study the limit, where applicants apply ``everywhere.'' For example, when $G$ is degenerate, in the ``dense'' network limit, the matching function converges to the (large market) urn-ball form of $f = 1-e^{-\theta}$.\footnote{To see this explicitly, note that $\lim_{\bar{d}_U \to \infty} \bar{d}_U\phi = \lim_{\bar{d}_U \to \infty} \theta(1-e^{-\bar{d}_U/\theta}) = \theta$. Note also that this job-finding probability is the limit of the expression in Corollary 3, when $V\to \infty$, holding $\theta$ fixed.} When $G$ is the exponential, (1) becomes $f = 1- (1+\bar{d}_U \phi)^{-1}$ and when $\bar{d}_U \to \infty$ it converges to $\theta/(1+\theta)$, the form assumed in \citet{KW1993}. We can show the result for any $G$ covered by \Cref{thm-1}: 

\begin{proposition}\label{prop-dense} When $\bar{d}_U \to \infty$ holding $\theta$ fixed, (1) converges to $f = 1-M_{\bar{d}_n}(-\theta)$, where $\bar{d}_n = \frac{\bar{d}}{\bar{d}_U}$ is applicant search intensity $\bar{d}$ normalized to have a mean of $1$.
\end{proposition}
\begin{proof}See \Cref{sec:app-D}.
\end{proof}

\Cref{prop-dense} applies to a whole family of \textit{dense networks}, which can be thought as a large-market generalization of the classic small-market urn-ball setup where applicants apply everywhere. This family is for example where \citet{Sephorah2025} focuses. This may be a very reasonable setup for multiple contexts, but it seems unlikely to be the case in the (large) job market, where applications arguably form a rather sparse network, with each applicant applying to only a handful of vacancies of the whole market $(\bar{d}_U < \infty)$. In any case, our setup applies equally well to a broad class of both dense and sparse networks. 

\textbf{Meetings when vacancies are ``abundant''}. Another limit case of interest is when the market tightness goes to infinity. This case is of interest primarily because, as the next proposition shows, a rather large and well-studied family of meeting technologies in the labor-search literature corresponds to precisely this case.\footnote{Papers using search technologies relating to this type include \citep{EK2010a, EK2010b, LVW2015, CGW2017, ACGV2020}.}

\begin{proposition} When $\theta \to \infty$ holding $\bar{d}_U$ fixed, (1) converges to $f = 1-M_{\bar{d}_n}(-\bar{d}_U)$, where $\bar{d}_n = \frac{\bar{d}}{\bar{d}_U}$ is applicant search intensity $\bar{d}$ normalized to have a mean of $1$.
\end{proposition}
\begin{proof}See \Cref{sec:app-D}.
\end{proof}

The expression $1-M_{\bar{d}_n}(-\bar{d}_U)$ is precisely the meeting probability given in eq. (17) of \citet{CGW2025}, and the degree distribution of sellers (applicants for us) follows a mixed Poisson with mean $\bar{d}_U$ as indicated in their eq. (16). Thus the two setups become equivalent in this limit case of ours.\footnote{\citet{CGW2025} denote the mean of the mixed Poisson ($\bar{d}_U$) by $\lambda$, as seen in their expression (16). Let us also note that they also refer to $\lambda$ as market tightness, which is absolutely true in their setup: when $d_j =1, \forall j$ as they assume, the accounting identity $\bar{d}_U=\theta \bar{d}_V$ indeed becomes $\bar{d}_U=\theta$.} 

Let us clarify. In \citet{CGW2025}, whenever a seller (applicant for us) connects to a buyer (vacancy for us), a meeting takes place. This is because in their setup each buyer links, by assumption, to a single seller, that is all buyers have a degree of $1$. Thus, it is impossible for multiple sellers to compete for the same buyer.  Our setup, in contrast, generally features competition on the buyer side, which is why each meeting has probability $\phi$ to yield a sale (offer for us). When vacancies become ``abundant'' with tightness $\theta \to \infty$, competition for buyers (vacancies) is eliminated, thus the two setups become equivalent.

Independently of whether one views this special case as each buyer (vacancy) linking to exactly one seller (applicant), or as market tightness $\theta \to \infty$, the outcome is the same: in this case, competition on the buyer (vacancy) side is muted. Its applicability depends on the specific application; it arguably seems less tenable in an uncoordinated labor-search context.

\textbf{The CES functional form}. Given (1), reverse-engineering suggests that a $G$ with a mgf of $M(t) = 1-\big(1+(-\bar{d}_U t)^{-\gamma}\big)^{-\frac{1}{\gamma}}, \gamma > 0$, defined over the non-positive values $t\leq 0$, would yield a job-finding probability of $\big(1+(\bar{d}_U\phi)^{-\gamma}\big)^{-\frac{1}{\gamma}}$. Then, in the dense network limit, when $\bar{d}_U \to \infty$, the job-finding probability would be of the appropriate form.\footnote{Assuming the mass of applicants is $U$ and the mass of vacancies $V$ the CES matching function is of the form $\big(U^{-\gamma}+V^{-\gamma}\big)^{-\frac{1}{\gamma}}$, with a job-finding probability $f = \frac{(U^{-\gamma}+V^{-\gamma})^{-\frac{1}{\gamma}}}{U}= \big(1+\theta^{-\gamma}\big)^{-\frac{1}{\gamma}}$ \citep[e.g.,][ch. 1]{dHRW2000, MS2015, P-NW2017}.}

We will come back to the exact representation of the CES functional form within our setup momentarily. Let us first notice that there is no well-known distribution with the above mgf. The next result, however, identifies a ``scaling'' condition of mean search intensity in the market, so that, at a 1st-order, the job-finding probability in every network in the family covered by \Cref{thm-1} can be thought ``as if'' it comes from a CES matching function.

\begin{proposition}\label{prop-scalling} Suppose the following condition holds between mean search intensity in the market, $\bar{d}_U$, and market tightness, $\theta = \frac{V}{U}$, for some constant $\gamma > 0$,
\[
\bar{d}_U = -\theta \ln\left(1+ \frac{1}{\theta}\ln\left(1 - \left(1+\theta^{-\gamma}\right)^{-\frac{1}{\gamma}}\right)\right) \tag{*}
\]
Then the job-finding probability 
\[
f = (1+\theta^{-\gamma})^{-\frac{1}{\gamma}}
\]
approximates the job-finding probability of \Cref{thm-1}, at a 1st-order. 
\end{proposition}
\begin{proof}
See \Cref{sec:app-D}.
\end{proof}

This result echoes \citet{Steve2007} who, using a specific queuing system primitive, derives the CES functional form when the system's parameters scale in a particular way. Whether such scaling conditions hold is of course an empirical question. In contrast to the queuing system primitive, our network of connections is observable in cross-sectional application data, thus (*) is directly testable.\footnote{Note that given the accounting identity of \Cref{lemma-acc-id}, according to which $\bar{d}_U = \bar{d}_V\theta$, condition (*) can be tested with data either on the applicant or the vacancy side.} In addition, conditional on it holding in the data, the proposition gives the CES form as a 1st-order approximation to a whole family of networks, rather than relying on a specific network.

Given our earlier discussion, let us note that \Cref{prop-scalling} applies to sparse networks. We now delve more in the case of dense networks and characterize a broad class of matching functions that can be microfounded within our setup.

\textbf{A representation result in dense networks}. \Cref{prop-dense} establishes that in the dense network limit, the job-finding probability is given by
\[
f(\theta) = 1-M_{\bar{d}_n}(-\theta)
\]
where $\bar{d}_n = \frac{\bar{d}}{\bar{d}_U}$ is applicant search intensity $\bar{d}$ normalized to have a mean of $1$, and $M(.)$ is the moment generating function of $\bar{d}_n$. From the definition of the mgf, it holds that
\[
M_{\bar{d}_n}(-\theta) = \int_0^{\infty}e^{-\bar{d}_n\theta}dG(\bar{d}_n)
\]
We recognize the right-hand side to be the \textit{Laplace transform} of the random variable $\bar{d}_n$ with $\theta$ representing the Laplace (or ``frequency'') domain. This observation alone gives a novel interpretation to the market tightness $\theta$, the key variable of frictional markets. 

The connection between the mgf and the Laplace transform is well-known in the literature. The following representation result \citep[e.g.,][ch. XIII.4]{Feller2} characterizes the functions that can be microfounded within the sub-class of dense networks of our setup: A function $m(\theta)$ on $(0, \infty)$ is the Laplace transform of a probability distribution $G$ iff it is \textit{completely monotone}, and $m(0)=1$. In this case, it will hold that 
\[
f(\theta) = 1 - m(\theta)
\]
In terms of definition, a function $m$ on $(0, \infty)$ is said to be completely monotone if it possesses derivatives $m^{(n)}$ of all orders, and $(-1)^n m^{(n)}(\theta) \geq 0, \ \theta >0$ \citep{Feller2}. From the above, when $m(\theta)$ is a Laplace transform, $f$ will satisfy all the ``standard'' properties of reduced-form matching technologies: $m(0)=1$ implies that $f(0)=0$, while complete monotonicity requires that $m' < 0, \ m'' > 0$, implying $f' > 0, \ f'' <0$.

Conversely, any increasing and concave function $f(\theta)$, for which $f(0)=0$, can be microfounded as a dense network limit of our setup, if it also has all its derivatives alternating in sign, and $f'(0)=1$. In that case $m(\theta)=1-f(\theta)$ is the Laplace transform of a CDF $G$ with mean $1$. The distribution of (normalized) search intensities $G$ can be found using the \textit{inverse Laplace transform}. 

As an application, let us consider the functions $m_1(\theta) = 1-(1+\theta^{-\gamma})^{-\frac{1}{\gamma}}, \ \gamma \in (0,1]$, and $m_2(\theta) = \left(1-\frac{1-e^{-\alpha/\theta}}{\alpha/\theta} \right)^\alpha, \ \alpha = 1,2,..., A < \infty$. 

Function $m_1$ can be shown to be completely monotone, so it corresponds to the Laplace transform of a CDF $G$; it can also be verified that $-m'(0)=1$, hence $G$ has a mean of $1$. The job-finding probability $f$ is of the CES functional form in this case. The CDF $G$ is not among the widely-known distributions, except for the case of $\gamma=1$, when $G$ is the exponential and the job-finding probability corresponds to that employed by \citet{KW1993}, as we have already seen. Function $m_2$ can be shown to not be completely monotone, so there is no distribution $G$, that could give a dense network microfoundation of the matching function in \citet{AGV2006}.

It is worth mentioning that our dense network microfoundation of the CES imposes a tighter restriction on $\gamma$, requiring it to be in $(0,1]$, while typically it is only assumed that $\gamma >0$.\footnote{In fact, \citet{dHRW2000} who introduced and popularized this functional form use a value of $\gamma = 1.27$ in their quantitative analysis. Even though we don't emphasize this there, $\gamma \leq 1$ is sufficient for our sparse network CES approximation result (\Cref{prop-scalling}) to be well-defined for all $\theta >0$ .} This comes from the stronger condition of complete monotonicity imposing restrictions on derivatives of all orders, while the reduced-form meeting technologies only impose restrictions on the first two derivatives.

\subsection{Comparative statics}
In addition to embedding multiple matching functions in the literature as special cases, our setup provides a theory of match efficacy. We now show how changes in the underlying distribution of search intensities $G$ affect the job-finding probability, i.e. how well the matching process works. The overarching theme of our analysis is that \textit{inequality} in search intensities is a key determinant of how well the matching process works. 

\textbf{The effects of dispersion}. Take the Taylor expansion of \Cref{prop-1st-ord} and add an extra term, yielding the 2nd-order approximation
\[
f \stackrel{\text{\tiny 2nd}}{\approx} 1 - e^{-\phi\bar{d}_U} - \Bigg[\frac{1}{2}\phi^2e^{-\phi\bar{d}_U} \Bigg] Var(\bar{d})
\]
As we can see, other things being equal, a higher variance has an unambiguously negative effect on the job-finding probability.\footnote{Given the special interest in the dense network limit, it is also worth noting that as $\bar{d}_U \to \infty$, if the variance stays bounded its effect will vanish and the job-finding probability will converge to the urn-ball value. This argument generalizes to moments of arbitrarily high order, as the Taylor series of (1) around $\bar{d}_U$ is given by $f=1-e^{-\phi\bar{d}_U}+\sum_{k=1}^\infty(-1)^{k+1}\frac{1}{k!}\phi^ke^{-\phi\bar{d}_U}\mathbb{E}[(\bar{d}-\bar{d}_U)^k]$.}

More generally, the next result  shows that pure increases in search intensity dispersion, in the sense of \textit{second-order stochastic dominance (SOSD)} shifts (mean-preserving spreads) in $G$, have unambiguously negative effects on matching. This is our second main result for the matching function. It also holds in the variants of our setup we analyze in the next two sections, making it a very robust finding in addition to its sharp characterization.

\begin{theorem}\label{thm-2} Suppose $G'$ is a mean-preserving spread of $G$, where $G, G'$ are applicant search intensity distributions satisfying the conditions of \Cref{thm-1}. Then, for the corresponding mean job-finding probabilities in the large market, it holds that 
\[
f' < f
\] 
\end{theorem}
\begin{proof}
See \Cref{sec:app-D}.
\end{proof}

The effects of dispersion in applicants' search intensities is absent from the standard matching function, which does not allow for effects beyond the mean \citep[e.g.,][ch. 5]{Pissa2000}. More micro-based approaches \citep[e.g.,][]{C-AZ2005, AGV2006} have focused on setups that preclude dispersion.

Having established the effect on matching of higher (pure) dispersion in search intensities, the natural next question is what is the effect of higher overall search intensity, $\bar{d}_U$.  

\textbf{The effects of higher overall search intensity}. It turns out that a higher $\bar{d}_U$ has generally an ambiguous effect on matching. Specifically, we show that it can have a negative effect when it is associated with an increase in inequality, in the sense of an increase in the \textit{Gini coefficient} of the distribution $G$.  

We first show that a necessary condition for an increase in $\bar{d}_U$ to have a negative effect, is for it to distort the shape of the distribution $G$: if search intensities scale up proportionally, as in a ``perfectly inflationary'' increase, the matching process unambiguously improves.

\begin{proposition} Define the proportional transformation to $\bar{d}$, $\bar{d}_\rho \equiv \rho \bar{d}$, where $\rho >0$ is a constant, and denote the corresponding mean job-finding probability of \Cref{thm-1} as $f(\rho)$. Then
\[
f(\rho) > f(1) \iff \rho > 1
\]
\end{proposition}
\begin{proof}
See \Cref{sec:app-D}.
\end{proof}

Next, we apply formula (1) for the cases of the applicant search intensity distribution $G$ being degenerate, gamma, Pareto, and uniform. We study the monotonicity of the job-finding probability $f$ as we increase the value of the mean $\bar{d}_U$ in ways that correspond to (parametric) \textit{first-order stochastic dominance (FOSD)} shifts of the distribution. \Cref{tab-1} summarizes our findings, while full details are provided in \Cref{sec:app-D}.

\begin{table}[ht]
\centering
\caption{Summary of parametric FOSD shifts}\label{tab-1}
\begin{tabular}{|l|c|c|}
\hline
\textbf{Distribution $G$} & \hspace{.5em}\textbf{Gini as $\bar{d}_U \uparrow$}\hspace{.5em} & \textbf{Job-finding probability $f$ as $\bar{d}_U \uparrow$} \\
\hline\hline
Degenerate & constant & monotonically increasing  \\\hline
Gamma & constant & monotonically increasing  \\\hline
Pareto & increasing & inverted-U shape\textsuperscript{2}  \\\hline
Uniform\textsuperscript{1} &  &  \\
\hspace{1em}Shift of type 1\hspace{.3em} & constant & monotonically increasing\textsuperscript{2}  \\
\hspace{1em}Shift of type 2\hspace{.3em} & decreasing & monotonically increasing\textsuperscript{2}  \\
\hspace{1em}Shift of type 3\hspace{.3em} & increasing & inverted-U shape\textsuperscript{2}  \\
\hline\hline
\end{tabular}
\vspace{0.5em}
\begin{minipage}{0.9\textwidth}
\textsuperscript{1}For the uniform we can perform shifts of different type. 
\textsuperscript{2}Shown graphically. 
\end{minipage}
\end{table}

The standard matching function is monotonically increasing in mean search intensity \citep[e.g.,][]{Pissa2000, EMR2015}.\footnote{The same is true for the generalized matching function of \citet{MPS2018}.} In contrast, our analysis shows that, when the higher mean search intensity is associated with an increase in the Gini coefficient of search intensities, a higher $\bar{d}_U$ can correspond to a lower job-finding probability. As seen for the example of the Pareto distribution in Figure 2, the reduction can be substantial.

\begin{figure}[!ht]
\centering
  \includegraphics[width=.9\linewidth]{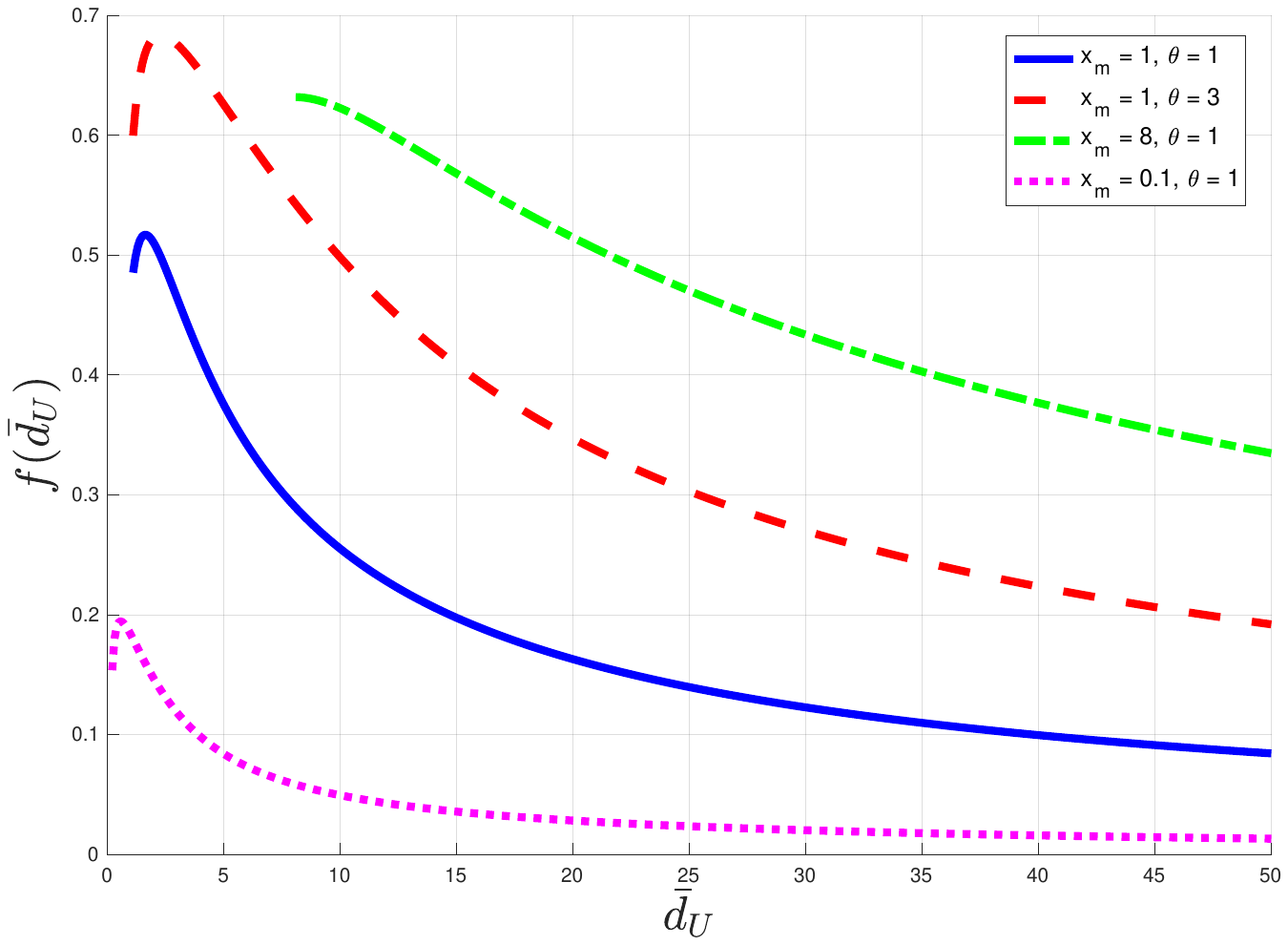}
\caption{Plot of job-finding probability when $G$ is Pareto for different pairs of $(x_m, \theta)$. $x_m$ is a scale parameter of the distribution, determining the left bound of its support; $\theta$ is the market tightness. Holding $x_m$ fixed, the shape parameter of the distribution $\alpha$ is varied from $\infty$ down to $1$, translating in $\bar{d}_U=x_m\left(1-\frac{1}{\alpha}\right)^{-1}$ varying from $x_m$ to $\infty$. As $\alpha \to \infty$, the Pareto becomes degenerate at $x_m$; so decreasing $\alpha$ (moving to the right of the figure) can be seen as a smooth departure from the benchmark Erdős-Rényi network.}
\end{figure} 
Such an inverted-U shape for the job-finding probability is featured in the setup of \citet{AGV2006}. It has appeared also in the social networks literature \citep{C-A2004, C-AZ2005}. Its existence beyond these fully symmetric cases and its connection to the Gini coefficient are novel.

\section{Aggregating vacancies to ``locations''}\label{sec:locations} 
All our analysis up to now has been at the applicant-vacancy level, without keeping track, for example, which vacancies belong to which firm, or perhaps to which recruiting agency. We are interested here in aggregating some of the vacancies to sets, or ``super-nodes.'' To conform with the prior literature we will be connecting to, such as \citet{Lagos2000}, \citet{Shimer2007} and \citet{Mortensen2009}, we will be referring to such sets of vacancies as ``locations.''

Denote the set of locations by $\mathcal{L}$, and assume there are $j=1..L$ of them. Each location has $0 \leq v_j < \infty$ vacancies, which we take to be exogenous; the total number of vacancies in the market is $V = \sum_{j=1}^L v_j$.

Applicants no longer apply to specific vacancies, but to locations. Analogously to before, we will assume that applicant-location linking probabilities are applicant-specific, i.e. $p_{ij}=p_i$. We will also assume that at each location $j$ all possible matches are formed, that is the probability for any $i$ connecting there to receive an offer is 
\[
\min\left\{1, \frac{v_j}{1+\tilde{d}_i}\right\}
\]
This assumption implies that within a location there are no coordination frictions, but across locations there are. In other words, the aggregation to locations, allows us to reduce the coordination frictions in the market. 

Note that for $v_j = 1, \forall j$ this setup features our previous setup as a special case. Before moving to the large market, let us see another special case of interest.

\textbf{The ``frictionless'' matching function}. When there are locations, the job-finding probability of any applicant is
\[
f_i = 1-\prod_{j=1}^L(1-p_i\phi_{j}(p_{-i}))
\]
where $\phi_{j}(p_{-i}) = \mathbb{E}_{\tilde{d}_i}\left[\min\left\{1, \frac{v_j}{1+\tilde{d}_i}\right\}\right]$, and $\tilde{d}_i\sim PB(\vec{p}_{-i})$.

When the whole market is a single location and the network is complete, i.e. $p_i = 1,\forall i$, the formula becomes
\[
f = \min\left\{1, \theta\right\}
\]
This job-finding probability corresponds to the ``frictionless'' matching function, in the sense that, by mere accounting, no more matches are possible than $U\cdot\min\left\{1, \theta\right\}$. \citet{Lagos2000} derives this matching function as an equilibrium object in an environment comparable to ours, when vacancies are optimally distributed across homogeneous locations.

\begin{definition}[Large market with locations] In the presence of locations, the large market is the limit of $U,L \to \infty$, while linking probabilities $p_i \to 0$. The limit value of market tightness $\frac{V}{U}$ is assumed to exist and be equal to $\theta$; the limit value of mean degree $p_i L$ is assumed to exist and be equal to $\bar{d}_i$, $\forall i$; the limit value of $\frac{\sum_i \bar{d}_i}{U}$ is assumed to exist and be equal to $\bar{d}_U$; the limit value of $\frac{V}{L}$ is also assumed to exist.
\end{definition}

The next result gives the analog of \Cref{thm-1}, in the setup with locations.

\begin{proposition}\label{thm-1_1} In a large market with locations, when applicant search intensities $\{\bar{d}_i\}_{i=1}^\infty$ are drawn i.i.d. from a distribution $G$, with a finite mean $\bar{d}_U$, support in a subset of $[0, \infty)$, and a moment-generating function (mgf) well-defined for non-positive values, and vacancies at locations $\{v_j\}_{j=1}^\infty$ are drawn i.i.d. from a distribution $H$ with support in $\mathbb{N}$ and mean $\bar{v}$, the mean job-finding probability is given by
\[
f = 1- M_{\bar{d}}(-\chi), 
\]
where $M_{\bar{d}}(t) \equiv \mathbb{E}[e^{\bar{d} t}]$ is the mgf of $\bar{d}_i, \forall i$; $\chi \equiv \mathbb{E}_{(v,\tilde{d})}\left[\min\left\{1, \frac{v}{1+\tilde{d}}\right\}\right]$, $\tilde{d} \sim Poisson(\bar{v}\cdot\bar{d}_U/\theta)$.
\end{proposition}
\begin{proof}
See \Cref{sec:app-E}.
\end{proof}
When $G$ is degenerate at $\bar{d}_U$, thus $\{\bar{d}_i\}_{i=1}^\infty$ are drawn from a Poisson with mean $\bar{d}_U$, and $H$ is a Poisson, the setup becomes that of \citet{Shimer2007} and \citet{Mortensen2009}. 

The next two results show that the effect of inequality in the sense of dispersion, as captured in \Cref{thm-2}, carries through also in the presence of locations. This is both for the distribution of degrees across applicants, and the distribution of vacancies across locations.

\begin{proposition} Suppose $G'$ is a mean-preserving spread of $G$, where $G, G'$ are applicant search intensity distributions satisfying the conditions of \Cref{thm-1_1}. Then, for the corresponding mean job-finding probabilities in the large market with locations, it holds that
\[
f' < f
\]
\end{proposition}
\begin{proof}
See \Cref{sec:app-E}.
\end{proof}
This result can compliment the findings on the compositional effects of the applicant pool as drivers of match efficacy \citep{BF2015, HS-W2018}.

\begin{proposition}\label{prop-MPS_H} Suppose $H'$ is a mean-preserving spread of $H$, where $H, H'$ are vacancy number distributions satisfying the conditions of \Cref{thm-1_1}. Then, for the corresponding mean job-finding probabilities in the large market with locations, it holds that
\[
f' \leq f
\]
\end{proposition}
\begin{proof}
See \Cref{sec:app-E}.
\end{proof}

Note that the random variable $\frac{v}{1+\tilde{d}}$ is the tightness of a (any) location, so any changes in the distribution of vacancies across locations maps into changes in the distribution of tightness. \Cref{prop-MPS_H} is thus a generalization of the dispersion effect in \citet{BF2015}. Variants of this result are found also in \citet{SSTV2014}  and \citet{HvR2020}; \citet{MR2018} conjecture that an analogous result underlies their environment in order to minimize geographic mismatch .

Empirically, it has indeed been documented that higher dispersion in the locations' tightness affects matching efficacy negatively \citep[e.g.,][]{SSTV2014, BF2015, MR2018, HvR2020}. The weak inequality in the proposition relates to the weak concavity coming from the $\min$ operator. This may be a contributing factor to why this type of dispersion has been found to be quantitatively limited in some cases \citep[e.g.,][]{SSTV2014, MR2018}

\section{Matching with heterogeneous job advertising}\label{sec:hetero-jobs} 
Starting with \citet{DFH2013}, and more recently \citet{GSV2018}, \citet{C-TGK2023}, \citet{MV2025}, the literature has highlighted the role of recruiting heterogeneity as a driver of match efficacy. In addition, in recent work, \citet{DavisSamaniego2} explicitly reject the distribution of applications per vacancy following a Poisson distribution in their dataset. 

With these facts in mind, in this section we switch the side of search heterogeneity, assuming link-formation probabilities are vacancy-specific, namely
\[
p_{ij} = p_j
\]

In this setup, $p_j$ is the probability of vacancy $j$ to attract an application by any given applicant. It can be interpreted as \textit{job-advertising intensity} and we will refer to it as such. To be able to compute the vacancy-filling probability, we also need to specify a protocol of how candidates choose among (potentially multiple) offers. Analogously to our foregoing analysis, we will assume candidates choose uniformly at random among multiple offers.\footnote{As in the Remark of \Cref{sec:setup}, we can assume each offer comes with a randomly drawn package of terms, and the applicant chooses the best offer.}

The mean vacancy degree in this case is given by $p_j U = \bar{d}_j$. The next result is the analog of \Cref{thm-1}, under job-advertising heterogeneity.

\begin{proposition}\label{thm-1_2} In a large market, when job-advertising intensities $\{\bar{d}_j\}_{j=1}^\infty$ are drawn i.i.d. from a distribution $\hat{G}$, with a finite mean $\bar{d}_V$, support in a subset of $[0, \infty)$, and a moment-generating function (mgf) well-defined for non-positive values, the mean vacancy-filling probability is given by
\[
q = \frac{1-e^{-\theta(1-M_{\bar{d}}(-1))}}{\theta}, \tag{2}
\]
where $M_{\bar{d}}(t) \equiv \mathbb{E}[e^{\bar{d} t}]$ is the mgf of $\bar{d}_j, \ \forall j$.
\end{proposition}
\begin{proof}
See \Cref{sec:app-F}.
\end{proof}
Analogously to the setup of \Cref{sec:hetero-apps}, vacancy degrees in this case follow a mixed Poisson, $P(d = k) = \int_0^\infty \frac{\bar{d}^k}{k!}e^{-\bar{d}}d\hat{G}(\bar{d})$, and applicant degrees a standard Poisson with mean $\bar{d}_U$.

The only case when formulas (1) and (2) both apply is the Erdős-Rényi network. For this case we can confirm their consistency,\footnote{See for example \citet{Pissa2000} p. 7, or \citet{P-NW2017} p. 4.} namely
\begin{corollary} When $\hat{G}$ is degenerate at $\bar{d}_V$ and $G$ is degenerate at $\bar{d}_U=\theta\bar{d}_V$, the job-finding probability of (1) and the vacancy-filling probability of (2) satisfy
\[
f=\theta q
\]
\end{corollary}
\begin{proof}
See \Cref{sec:app-F}.
\end{proof}
Finally, we can show that the negative effect of inequality in the sense of dispersion, as captured in \Cref{thm-2}, holds on the side of vacancies as well.
\begin{proposition} Suppose $\hat{G}'$ is a mean-preserving spread of $\hat{G}$, where $\hat{G}$ and $\hat{G}'$ are job-advertising intensity distributions satisfying the conditions of \Cref{thm-1_2}. Then, for the corresponding mean vacancy-filling probabilities in the large market, it holds that
\[
q' < q
\]
\end{proposition}
\begin{proof}
See \Cref{sec:app-F}.
\end{proof}

\section{Concluding remarks}\label{sec:conclusion}
This paper establishes the significance of the applicant-vacancy (random) network as an intermediate, granular object for opening the ``black box'' of the complex matching process. How these networks look in the data, how they are formed, how they evolve over the business cycle, how they respond to policy, remain to be studied systematically.


We have already mentioned the Pareto distribution as a candidate distribution from which search intensities can be drawn. Given the growing literature on the influence of social networks on labor-market outcomes \citep[see, for example, the review of][]{LB2016}, it is worth noting that such fat-tailed degree distributions can indeed be outcomes of ``preferential attachment'' network models \citep[e.g.,][ch. 5]{Jackson2010} relating to people's popularity.

One dimension of structure that we have not incorporated into our analysis is the notion of ``closeness'' between vacancies (or ``locations''), where ``distance'' may refer to any relevant dimension (e.g. sector, geography, shared social ties). For any dimension, one would expect correlation patterns to emerge among vacancies that are ``close'' to each other. Suppose applicant $i$ connects to two jobs---say $j$ and $j'$. Conditional on applicant $i'$ connecting to job $j$, they may also be more likely to connect to job $j'$.\footnote{Relatedly, in the case of ``locations,'' we abstract from the fact that the number of vacancies per location will plausibly correlate positively with job-advertising intensity at that location.} 

This relates to the fundamental notion of clustering (or transitivity) in the network literature \citep[e.g.,][]{Newman2003}, and it is absent from our random network structures, for which the probabilities of each link are independent. We think this is interesting to refine in future work; such a setup could exploit more fully the generality of the setup of $p_{ij}$ we laid out in the beginning. The theoretical challenge is non-negligible, as network models featuring clustering are less well studied in the literature. Apart from its realism, clustering is expected to matter quantitatively, as we would expect coordination failures be exacerbated in its presence: To put it simply, clustering implies that the same people compete for the same jobs. 

Networks can also guide refined welfare types of exercises,\footnote{Finding whether constrained efficiency is achieved is perhaps the most classic exercise in the search literature \citep[e.g.,][]{Moen1997, Shi2005, AGV2006, Kircher2009}.} where the planner takes into account the whole network of connections, thus finding the constrained optimum by eliminating only coordination frictions. This is the ``max-flow'' problem, a classic problem in the transportation networks literature. The max-flow problem, contrary to our interest in this paper, requires coordination in the assignment of matches. It is precisely the problem \citet{GH2017} show how to solve analytically.

Regarding how offers are made, we made the modeling choice that each firm makes an offer to a single---the best---candidate when qualities are drawn from the same distribution for all applicants. Such formalization precludes, firstly, the possibility of vertical heterogeneity across applicants, and secondly, the possibility that a firm does not make any offer, or that they extend multiple rounds of offers, when their preferred candidates reject their offer. 

Our environment shares both these aspects with much of the theoretical literature.\footnote{For a comprehensive (guided) tour to the directed search literature see \citet{Wetal2021}.} A notable exception regarding the second aspect is \citet{Kircher2009}, who makes the polar assumption of a vacancy being offered to (potentially) all applicants, if earlier offers are rejected.\footnote{Allowing for potentially all rounds of offers is equivalent to assuming the matching outcome is ``stable.''} Arguably, neither assumption is superior a priori: How far down their applicant list firms go, or even whether they extend any offer in the first place, is an inherently dynamic decision, balancing the quality of their received applications today and the qualities they expect to receive tomorrow. We have thus little to say on it in a static environment. An analogous argument applies to whether applicants accept any of the offers they receive.

Indeed such endogenous acceptance/rejection cutoffs have been shown to matter in the dynamic setups of \citet{Petr2014}, \citet{Wolt2018} and \citet{MM2020} to explain both cyclical and secular patterns in labor market data. \citet{BF2015} and \citet{HS-W2018},  discussed earlier, document the relevance of vertical heterogeneity when applicants differ in their degree of  hirability and locations (sub-markets) in their degree of selectivity.

Granted the empirical relevance of these aspects, our intention was not to create a setup that captures all dimensions of reality, rather to uncover the implications of search patterns---as these are embedded in the random network---for the  matching function. This paper offers a minimal, novel, and useful setup to do precisely that. In this sense, random networks offer a natural formalization of what \citet{Stigler1961, Stigler1962} termed ``search.''

Let us conclude by noting that even though our terminology and most of our references are for the labor market, our analysis is not context-specific, and could be applied to any market where search frictions have been argued to play a role.\footnote{Our earlier comment on the plausibility of the Pareto is perhaps more directly applicable in a consumption context, where ``applicants'' are buyers and their search intensity depends on their income.} It just seemed most natural to work in this context, which has almost become inseparable from the matching function. 


\begin{appendices}
\crefalias{section}{appsec}
\setcounter{theorem}{0}
\setcounter{lemma}{0}
\setcounter{proposition}{0}
\setcounter{corollary}{0}
\section{Introduction}\label{sec:app-A} 
\textbf{Elementary properties of the Poisson binomial distribution}. A sum of $n$ independent Bernoulli trials, with (potentially) different probabilities of ``success'' $p_i$ each, is a random variable $S_n$ said to follow the “Poisson binomial” distribution. The binomial is thus the special case of the Poisson binomial, when all trials have the same probability of success ($p_i=p,\ \forall i=1..n$). As with the binomial, $S_n$ counts the number of successes among the $n$ trials. 

It is parameterized by the success probabilities of the $n$ trials $\vec{p}=(p_1,p_2,..,p_n)$; we will denote it by $PB(\vec{p})$; its mean is $\sum_{i=1}^n p_i$; when $n\to \infty$ (appropriately), the Poisson binomial converges (in distribution) to a Poisson with the same mean \citep[e.g.,][p. 282]{Feller}. 

\section{The setup}\label{sec:app-B} 
\begin{remark*} Suppose a vacancy $j$ receives $d_j$ applications, and each is associated with a measure of quality, drawn independently from some underlying distribution. Then, the probability that any given applicant is the most qualified for the position is $\frac{1}{d_j}$.
\end{remark*}
\begin{proof} Denote the applicant qualities corresponding to the $d_j$ applicants by $A_1, A_2, ..., A_{d_j}$. The probability that applicant $1$ is offered the job, is the probability their quality is the highest, i.e
\[
Pr\{A_1 > A_2, A_3,..,A_{d_j}\}
\]
Since these are i.i.d, by symmetry it follows that
\begin{align*}
Pr\{A_1 > A_2, A_3,..,A_{d_j}\} &= Pr\{A_2 > A_1, A_3,..,A_{d_j}\}\\
								 &= Pr\{A_3 > A_1, A_2,..,A_{d_j}\}\\
								 &= ... \\
								 &= Pr\{A_{d_j} > A_1, A_2, A_3,..,A_{d_j-1}\}
\end{align*}
But these are mutually exclusive and exhaustive events, thus they sum to $1$. As these are $d_j$ events, it follows that 
\[
Pr\{A_i > A_1, A_2,.., A_{i-1}, A_{i+1}, A_{d_j}\} = \frac{1}{d_j}, \ \forall i=1..d_j
\]
\end{proof}
\section{Matching with heterogeneous search intensity}\label{sec:app-C} 
\begin{lemma} When linking probabilities are given by $p_{ij}=p_i$, applicant degrees $d_i \equiv \sum_j Y_{ij}$ are independent and each follows a binomial distribution, $d_i \sim Bin(V, p_i)$. Vacancy degrees $d_j \equiv \sum_i Y_{ij}$ are i.i.d. following a Poisson binomial distribution,\footnote{A random variable is said to follow the ``Poisson binomial'' distribution when it is a sum of independent Bernoulli trials with (potentially) different probabilities of success; it is a generalization of the binomial distribution. We give its relevant properties in \Cref{sec:app-A}.} $d_j \sim PB(\vec{p})$.
\end{lemma}
\begin{proof} Any two $d_i, d_{i'}$ for $i\neq i'$ are independent as they don't share any of the independent Bernoulli random variables they sum over. The same applies to any two $d_j, d_{j'}$ for $j\neq j'$.

By assumption, fixing $i$, all $Y_{ij}, \forall j=1..V$ have the same probability of success, $p_i$. It follows that $d_i \sim Bin(V, p_i)$; this is the case for all $i=1..U$.

Also by assumption, fixing $j$, each $Y_{ij}, \forall i=1..U$ has a probability of success, $p_i$. It follows that $d_j \sim PB(\vec{p})$; this is the case for all $j=1..V$.
\end{proof}


\begin{corollary}When linking probabilities are given by $p_{ij}=p_i$, applicant $i$'s indirect degrees at any vacancy are i.i.d. following the Poisson binomial distribution $PB(\vec{p}_{-i})$, where $\vec{p}_{-i} = (p_1, p_2,...,p_{i-1},p_{i+1},...,p_U)$. We will denote $\tilde{d}_{i} \equiv \sum_{k \neq i} Y_{kj}$, $\forall j$.
\end{corollary}
\begin{proof}Conditional on linking to a vacancy $j$, the applicant's indirect degree at that vacancy is $\sum_{k \neq i} Y_{kj}$. By the same argument as in Lemma 1, this random variable follows a Poisson binomial $PB(\vec{p}_{-i})$. This is the case for all $j=1..V$, thus $\tilde{d}_{i} \sim PB(\vec{p}_{-i})$. Independence $\forall j$ also follows by the same argument as in Lemma 1.
\end{proof}

\begin{lemma}When linking probabilities are given by $p_{ij}=p_i$, applicant $i$'s job-finding probability is given by
\[
f_i = 1 - \big(1-p_i\phi\big(\vec{p}_{-i}\big)\big)^V
\]
where $\phi\big(\vec{p}_{-i}\big) \equiv \mathbb{E}\left[\frac{1}{1+\tilde{d}_i}\right]$, and $\tilde{d}_i \sim PB(\vec{p}_{-i})$.
\end{lemma}
\begin{proof} To give the proof as clean as possible, it is useful to define the notion of a node's neighborhood: Applicant $i$'s neighborhood, $N_i$, is the set of vacancies the applicant links to.  

We will also use the notation $\tilde{d}_{ij}$ to explicitly denote applicant $i$'s indirect degree at a particular vacancy $j$, that is $\tilde{d}_{ij} \equiv \sum_{k\neq i}Y_{kj}$.
\begin{align*}
f_i &= 1 - \mathbb{E}_{N_i}\left\{\mathbb{E}_{(\tilde{d}_{ij})_{j\in N_i}} \left\{\prod_{j\in N_i} \left(1 - \frac{1}{1+\tilde{d}_{ij}}\right)\right\}\right\}\\
	&= 1 - \sum_{N_i}P(N_i)\left\{\mathbb{E}_{(\tilde{d}_{ij})_{j\in N_i}} \left\{\prod_{j\in N_i} \left(1 - \frac{1}{1+\tilde{d}_{ij}}\right)\right\}\right\}\\
	&= 1 - \sum_{N_i}P(N_i) \left\{\prod_{j\in N_i} \left(1 - \mathbb{E}_{\tilde{d}_{ij}}\frac{1}{1+\tilde{d}_{ij}}\right) \right\}\\	
	&= 1 - \sum_{k=0}^V P\big(|N_i| = k\big) \left(1 - \mathbb{E}_{\tilde{d}_{i}}\left[\frac{1}{1+\tilde{d}_i}\right]\right)^k\\
	&= 1 - \sum_{k=0}^V P\big(d_i = k\big) \left(1 - \phi(\vec{p}_{-i})\right)^k \tag{**}
\end{align*}
The 1st line is notation for the definition of $f_i$; $P(N_i)$ in the 2nd line denotes the probability that $i$'s neighborhood is the set $N_i$; the 3rd line inter-changes the expectation and product operators, following from the independence of indirect degrees at all vacancies (Corollary 1); the fourth line follows from the fact that indirect degrees across vacancies are identically distributed (Corollary 1), hence we drop the explicit $\tilde{d}_{ij}$ notation and denote $\tilde{d}_{ij} = \tilde{d}_{i}$, $\forall j$; the last line follows from the observation that the size of an applicant's neighborhood is their degree, and denotes $\phi(\vec{p}_{-i}) \equiv \mathbb{E}\left[\frac{1}{1+\tilde{d}_i}\right]$.

From Corollary 1 we know $\tilde{d}_i \sim PB(\vec{p}_{-i})$ so we can evaluate $\phi(\vec{p}_{-i})$, and from Lemma 1 we know that $d_i \sim Bin(V, p_i)$ hence we know what $P\big(d_i = k\big)$ are. Finally, we recognize that (**) is $1$ minus $d_i$'s probability generating function (pgf) evaluated at $1-\phi(\vec{p}_{-i})$. The binomial has a known pgf, which, evaluated at $1-\phi(\vec{p}_{-i})$ yields
\begin{align*}
f_i &= 1-\big(1-p_i + p_i\big(1-\phi(\vec{p}_{-i})\big)\big)^V\\
    &= 1-\big(1-p_i\phi(\vec{p}_{-i})\big)^V
\end{align*}
\end{proof}

\begin{corollary}In the special case when linking probabilities are given by $p_i = 1, \forall i$, the job-finding probability is given by
\[
f = 1 - \left(1-\frac{1}{U}\right)^V
\]
\end{corollary}
\begin{proof}When $p_i=1$, each applicant links to all vacancies. Then $\tilde{d}_i = U-1$, and $\phi=\frac{1}{U}$, $\forall i$. It follows that the formula of Lemma 2 becomes $f = 1 - \left(1-\frac{1}{U}\right)^V$.
\end{proof}

\begin{lemma}The market-level average of mean applicant degrees $\bar{d}_U \equiv \frac{\sum_i \bar{d}_i}{U}$, the vacancies' mean degree $\bar{d}_V$, and the market tightness $\theta = \frac{V}{U}$, satisfy $\bar{d}_U = \theta \bar{d}_V$.
\end{lemma}
\begin{proof} For any market size we have
\begin{align*}
\bar{d}_U &= \frac{\sum_i \bar{d}_i}{U}\\
	&= \frac{\sum_i p_iV}{U}\\
	&= \frac{\bar{d}_V V}{U}\\
	&= \bar{d}_V \theta
\end{align*}
\end{proof}

\begin{lemma}In a large market, applicant degrees are independent each following a Poisson distribution with mean $\bar{d}_i$. Vacancy degrees are i.i.d. following a Poisson with mean $\bar{d}_V$.
\end{lemma}
\begin{proof}Applicant degrees are independent binomials, $Bin(V, p_i)$ for each $i$. In the large market we have $p_i\to 0$ and $p_iV \to \bar{d}_i$, $\forall i$ as $V\to \infty$. Under these conditions each binomial is known to converge to a Poisson with mean $\bar{d}_i$ \citep[e.g.,][p. 280]{Feller}. Independence follows from the independence of the binomials.

Vacancy degrees are i.i.d following $PB(\vec{p})$, with mean $\sum_i p_i$. We have that 
\begin{align*}
\sum_i p_i &= \sum_i \frac{\bar{d}_i}{V}\\
		   &= \theta^{-1}\frac{\sum_i \bar{d}_i}{U}
\end{align*}  
Having assumed that $\theta=\frac{V}{U}$ stays constant, and $\frac{\sum_i \bar{d}_i}{U} \to \bar{d}_U$, it follows that as $p_i \to 0, \ \forall i$, $\sum_i p_i \to \theta^{-1}\bar{d}_U$ which are the conditions for the Poisson binomial to converge to a Poisson with mean $\theta^{-1}\bar{d}_U$ \citep[e.g.,][p. 282]{Feller}. From Lemma 3, we have $\theta^{-1}\bar{d}_U = \bar{d}_V$.
\end{proof}

\begin{lemma}In a large market, when search intensities are given by $\{\bar{d}_i\}_{i=1}^{\infty}$, applicant $i$'s job-finding probability is given by
\[
f_i = 1-e^{-\bar{d}_i\phi}
\]
where $\phi \equiv \frac{1-e^{-\bar{d}_V}}{\bar{d}_V}$. 
\end{lemma}
\begin{proof}Suppose that $\phi(\vec{p}_{-i}) \to \phi, \forall i$, as we will show is the case. Then indeed
\begin{align*}
f_i &= 1- \big(1-p_i\phi(\vec{p}_{-i})\big)^V\\
	&= 1- \big(1-\frac{\bar{d}_i\phi(\vec{p}_{-i})}{V}\big)^V\\
	&\to 1-e^{-\bar{d}_i\phi}
\end{align*}
Now, the indirect degree distribution of applicant $i$ is $PB(\vec{p}_{-i})$, a Poisson binomial with mean $\sum_{k=1}^U p_k - p_i$. In the large market, when $p_i \to 0, \forall i$, clearly this mean converges to $\sum_i p_i =\bar{d}_V$. That is, indirect degrees (as vacancy degrees) become i.i.d. following a Poisson with mean $\bar{d}_V$. This implies that indeed $\phi(\vec{p}_{-i}) \to \phi, \forall i$. It still remains to find this value $\phi$.

We have $\phi \equiv \mathbb{E}\left[\frac{1}{1+\tilde{d}}\right]$, where in the large market we established that $\tilde{d} \sim Pois(\bar{d}_V)$. Then
\begin{align*}
\phi &= \mathbb{E}\left[\frac{1}{1+\tilde{d}}\right]\\
     &= \sum_{k=0}^\infty \frac{e^{-\bar{d}_V}(\bar{d}_V)^k}{k!}\frac{1}{1+k}\\
     &= \frac{1}{\bar{d}_V}\sum_{k=0}^\infty \frac{e^{-\bar{d}_V}(\bar{d}_V)^{k+1}}{(k+1)!}\\
     &= \frac{1}{\bar{d}_V}\sum_{k=1}^\infty \frac{e^{-\bar{d}_V}(\bar{d}_V)^{k}}{k!}\\
     &= \frac{1-P(\tilde{d}=0)}{\bar{d}_V}\\
     &= \frac{1-e^{-\bar{d}_V}}{\bar{d}_V}
\end{align*}
\end{proof}

\begin{theorem}In a large market, when applicant search intensities $\{\bar{d}_i\}_{i=1}^\infty$ are drawn i.i.d. from a distribution $G$ with a finite mean $\bar{d}_U$, support in a subset of $[0, \infty)$, and a moment-generating function (mgf) well-defined for non-positive values, the mean job-finding probability is given by
\[
f(\theta; \ G) = 1- M_{\bar{d}}(-\phi), \tag{1}
\]
where $M_{\bar{d}}(t) \equiv \mathbb{E}[e^{\bar{d} t}]$ is the moment-generating function (mgf) of $\bar{d}_i, \ \forall i$, and $\phi \equiv \frac{1-e^{-\bar{d}_U/\theta}}{\bar{d}_U/\theta}$.
\end{theorem}
\begin{proof}Conditional on a realized search intensity $\bar{d}_i$, the job finding probability is given by the expression in Lemma 5. It follows, that the unconditional job-finding probability, when $\bar{d}_i$ is drawn from $G$ is
\begin{align*}
f(\theta; G) &= \mathbb{E}_{\bar{d}_i}[f_i]\\
  &= 1-\mathbb{E}_{\bar{d}_i}\left[e^{-\bar{d}_i\phi}\right]\\
  &= 1-M_{\bar{d}_i}(-\phi)
\end{align*}
The conditions of Theorem 1 guarantee that everything is well-defined. We drop the $i$ for notational convenience.
\end{proof} 

\begin{proposition} The function $f(\theta; G)$ of Theorem 1 is increasing and concave in $\theta$ and satisfies $f(0; G)=0$.
\end{proposition} 

\begin{proof}
We have 
\[
f(\theta; G) = 1-\mathbb{E}_{\bar{d}}\left[e^{-\bar{d}\phi}\right]
\]
where $\phi = \frac{1-e^{-\bar{d}_U/\theta}}{\bar{d}_U/\theta}$. 

The fact that $f(0; G)=0$ follows from $\lim_{\theta \to 0^{+}} \frac{1-e^{-\bar{d}_U/\theta}}{\bar{d}_U/\theta} = 0$, as is readily verifiable.

To show monotonicity and concavity, we take the 1st and 2nd derivatives wrt to $\theta$:
\begin{align*}
\frac{\partial f(\theta; G)}{\partial\theta} &= \phi'\mathbb{E}_{\bar{d}}\left[\bar{d} e^{-\bar{d}\phi}\right]\\
\frac{\partial^2 f(\theta; G)}{\partial\theta^2} &= \phi'' \mathbb{E}_{\bar{d}}\left[\bar{d} e^{-\bar{d}\phi}\right] - \left(\phi'\right)^2\mathbb{E}_{\bar{d}}\left[\bar{d}^2 e^{-\bar{d}\phi}\right]\\
\end{align*}
If $\phi' > 0$ and $\phi'' < 0$, then
\begin{align*}
\frac{\partial f(\theta; G)}{\partial\theta} &>0\\
\frac{\partial^2 f(\theta; G)}{\partial\theta^2} &<0\\
\end{align*}
establishing the result. Indeed $\phi$ is increasing and concave in $\theta$ as we show now. For compactness of notation, let us use $\bar{d}_V = \bar{d}_U / \theta$, so that
\[
\phi = \frac{1 - e^{-\bar{d}_V}}{\bar{d}_V}
\]

Then
\begin{align*}
\phi' &\equiv \frac{d\phi}{d\theta} = \left[\frac{1-e^{-\bar{d}_V}}{\bar{d}_V}\right]'\frac{\partial \bar{d}_V}{\partial \theta}\\
			&= \frac{e^{-\bar{d}_V} \bar{d}_V - (1 - e^{-\bar{d}_V})}{\bar{d}_V^2}\frac{\partial \bar{d}_V}{\partial \theta}\\ 
			&= \frac{e^{-\bar{d}_V}(1+\bar{d}_V)-1}{\bar{d}_V^2}\frac{\partial \bar{d}_V}{\partial \theta}\\
			&= \frac{e^{-\bar{d}_V}(1+\bar{d}_V)-1}{\bar{d}_V^2} \left(-\frac{\bar{d}_U}{\theta^2}\right)\\ 
			&= \frac{1 - e^{-\bar{d}_V}(1+\bar{d}_V)}{\bar{d}_V\,\theta}
\end{align*}
It remains to show the numerator is positive for $\bar{d}_V > 0$. Define
\[
g(\bar{d}_V) = 1 - e^{-\bar{d}_V}(1+\bar{d}_V)
\]
It holds that $g(0)=0$ and $g'(\bar{d}_V) = \bar{d}_V\,e^{-\bar{d}_V} >0$, thus $\phi' >0$.

To establish the sign of $\phi''$, let us observe that the denominator of $\phi'$ is $\bar{d}_V \theta = \bar{d}_U$, which is constant in $\theta$. Therefore 
\begin{align*}
\phi'' &= \frac{1}{\bar{d}_U} \frac{d [g(\bar{d}_V)]}{d\theta}\\
							&= \frac{1}{\bar{d}_U} g'(\bar{d}_V) \frac{\partial\bar{d}_V}{\partial\theta}\\ 
							&= -\frac{\bar{d}_V\,e^{-\bar{d}_V}}{\theta^2}
\end{align*}
Hence $\phi'' <0$, completing the proof.
\end{proof}

\section{Implications of heterogeneous search intensity}\label{sec:app-D} 
\begin{proposition}Taking a 1st-order Taylor expansion with respect to $\bar{d}$ around $\bar{d}_U$, the job-finding probability of Theorem 1 becomes
\[
f \stackrel{\text{\tiny 1st}}{\approx} 1-e^{-\bar{d}_U\phi}
\]
\end{proposition}
\begin{proof}We take the 1st-order Taylor expansion of the function in Lemma 5, with respect to $\bar{d}_i$, around $\bar{d}_U$:
\[
f_i \approx 1-e^{-\bar{d}_U\phi} +\phi e^{-\bar{d}_U\phi}(\bar{d}_i - \bar{d}_U)
\]
The expectation of that yields
\[
f \approx 1-e^{-\bar{d}_U\phi}
\]
where, in the large market $\mathbb{E}[\bar{d}_i] = \bar{d}_U$. 
\end{proof}

\begin{proposition} When $\bar{d}_U \to \infty$ holding $\theta$ fixed, (1) converges to $f = 1-M_{\bar{d}_n}(-\theta)$, where $\bar{d}_n = \frac{\bar{d}}{\bar{d}_U}$ is applicant search intensity $\bar{d}$ normalized to have a mean of $1$.
\end{proposition}
\begin{proof}We have that
\begin{align*}
f &= 1-M_{\bar{d}}(-\phi)\\
  &= 1-\int_0^\infty e^{-\phi \bar{d}}dG(\bar{d})
\end{align*}
We do the change of variable $\bar{d}_n = \frac{\bar{d}}{\bar{d}_U}$, where $\bar{d}_n$ is $\bar{d}$ normalized to have a mean of $1$. It follows that
\begin{align*}
f &= 1-\int_0^\infty e^{-\phi \bar{d}}dG(\bar{d})\\
  &= 1-\int_0^\infty e^{-\phi\bar{d}_U \bar{d}_n}dG(\bar{d}_n)\\
  &\to 1-\int_0^\infty e^{-\theta \bar{d}_n}dG(\bar{d}_n) 
\end{align*}
Where the last line follows from $\lim_{\bar{d}_U \to \infty} \bar{d}_U\phi = \lim_{\bar{d}_U \to \infty} \theta(1-e^{-\bar{d}_U/\theta}) = \theta$.
\end{proof}

\begin{proposition} When $\theta \to \infty$ holding $\bar{d}_U$ fixed, (1) converges to $f = 1-M_{\bar{d}_n}(-\bar{d}_U)$, where $\bar{d}_n = \frac{\bar{d}}{\bar{d}_U}$ is applicant search intensity $\bar{d}$ normalized to have a mean of $1$.
\end{proposition}
\begin{proof} We proceed as in Proposition 3. We do the change of variable $\bar{d}_n = \frac{\bar{d}}{\bar{d}_U}$, where $\bar{d}_n$ is $\bar{d}$ normalized to have a mean of $1$. It follows that
\begin{align*}
f &= 1-\int_0^\infty e^{-\phi \bar{d}}dG(\bar{d})\\
  &= 1-\int_0^\infty e^{-\phi\bar{d}_U \bar{d}_n}dG(\bar{d}_n)\\
  &\to 1-\int_0^\infty e^{-\bar{d}_U \bar{d}_n}dG(\bar{d}_n) 
\end{align*}
Where the last line follows from $\lim_{\theta \to \infty} \phi = 1$, applying L'Hôpital's rule. 
\end{proof}

\begin{proposition} Suppose the following condition holds between mean search intensity in the market, $\bar{d}_U$, and market tightness, $\theta = \frac{V}{U}$, for some constant $\gamma > 0$,
\[
\bar{d}_U = -\theta \ln\left(1+ \frac{1}{\theta}\ln\left(1 - \left(1+\theta^{-\gamma}\right)^{-\frac{1}{\gamma}}\right)\right) \tag{*}
\]
Then the job-finding probability 
\[
f = (1+\theta^{-\gamma})^{-\frac{1}{\gamma}}
\]
approximates the job-finding probability of Theorem 1, at a 1st-order.  
\end{proposition}
\begin{proof}
Follows directly from Proposition 2 by replacing $\bar{d}_U$ from (*).
\end{proof}

\begin{theorem}Suppose $G'$ is a mean-preserving spread of $G$, where $G, G'$ are applicant search intensity distributions satisfying the conditions of Theorem 1. Then, for the corresponding mean job-finding probabilities in the large market, it holds that
\[
f' < f
\]
\end{theorem}
\begin{proof} Note that in the expression of Lemma 5,
\[
f_i = 1 - e^{-\bar{d}_i \phi}
\]
$\phi$ is a function only of the mean, thus it stays the same in a MPS. $f_i$ is an increasing and concave function of $\bar{d}_i$, since
\begin{align*}
\frac{\partial}{\partial \bar{d}_i} &= \phi e^{-\bar{d}_i \phi} > 0\\
\frac{\partial^2}{\partial \bar{d}_i^2} &= -\phi^2 e^{-\bar{d}_i \phi} < 0
\end{align*}
It follows \citep[][proposition 6.D.2]{M-CWG1995} that  
\begin{align*}
\int \big(1 - e^{-\bar{d} \phi}\big) dG'(\bar{d}) &< \int \big(1 - e^{-\bar{d} \phi}\big) dG(\bar{d}) \ \text{or}\\
f' &< f
\end{align*}
\end{proof}

\begin{proposition}Define the proportional transformation to $\bar{d}$, $\bar{d}_\rho \equiv \rho \bar{d}$, where $\rho >0$ is a constant, and denote the corresponding mean job-finding probability of Theorem 1 as $f(\rho)$. Then
\[
f(\rho) > f(1) \iff \rho > 1
\]
\end{proposition}
\begin{proof}
From the expression of Lemma 5, replacing $\bar{d}_i$ with $\rho\bar{d}_i$, we get
\[
f_i(\rho) =1 - e^{-\frac{\bar{d}_i}{\bar{d}_U}\theta\left(1-e^{-\rho \bar{d}_U/\theta}\right)}
\]
which is monotonically increasing in $\rho$. It follows that $f_i(\rho) > f_i(1)$ iff $\rho >1$, hence also after integrating
\[
f(\rho) > f(1) \iff \rho > 1
\]
\end{proof}

\textbf{The effects of higher overall search intensity}. Let us look at each of the cases in Table 1 of the main text, one at a time.

\underline{Case 1}: $G$ is the degenerate distribution. Formula (1) becomes $f = 1-e^{-\bar{d}_U\phi}$. $\bar{d}_U\phi$ is readily seen to be increasing in $\bar{d}_U$, hence $f$ is increasing in $\bar{d}_U$. An increase in $\bar{d}_U$ corresponds to a FOSD shift of the degenerate distribution. The Gini of the degenerate is equal to $0$, and hence constant as $\bar{d}_U$ changes.

\underline{Case 2}: $G$ is the gamma distribution. Formula (1) becomes $f = 1-(1+\phi\bar{d}_Uk^{-1})^{-k}$, where $k >0$ is the shape parameter, and $\bar{d}_Uk^{-1}$ the scale parameter. Given that $\bar{d}_U\phi$ is increasing in $\bar{d}_U$, $f$ is readily seen to be increasing in $\bar{d}_U$, when $k$ is held fixed. An increase in $\bar{d}_U$ corresponds to a FOSD shift of the gamma distribution; this follows from the lower incomplete gamma function appearing in the CDF of the gamma distribution being known to be decreasing in its second argument, implying the CDF moves to the right as $\bar{d}_U$ goes up. The Gini of the gamma is equal to only be a function of the scale parameter $k$, and hence constant as $\bar{d}_U$ changes.

\underline{Case 3}: $G$ is the Pareto distribution. Formula (1) becomes $f = 1-\alpha(x_m\phi)^\alpha \Gamma(-\alpha, x_m\phi)$, where $x_m>0$ is the scale parameter corresponding to the minimum value and $\alpha>0$ the shape parameter; for $\alpha>1$ the mean is finite and given by $\bar{d}_U = x_m\left(1-\frac{1}{\alpha}\right)^{-1}$, so we restrict attention to this range; $\Gamma(\cdot,\cdot)$ is the incomplete gamma function. As we decrease $\alpha$, holding $x_m$ fixed, we increase $\bar{d}_U$; it can be seen graphically (Figure 2 of the main text), that for any given $\theta$, $f$ exhibits an inverted-U shape in this case. It can also be seen graphically that a decrease in $\alpha$ corresponds to a shift to the right of the Pareto CDF, and hence a FOSD shift of it. The Gini of the Pareto is equal to $\frac{1}{2\alpha-1}$, and is hence decreasing in $\alpha$ and increasing in $\bar{d}_U$.

\underline{Case 4}: $G$ is the uniform distribution. Formula (1) becomes $f = 1+\frac{e^{-b\phi}-e^{-a\phi}}{\phi(b-a)}$, where $(a,b)$ is the support of $G$ and $\bar{d}_U = \frac{a+b}{2}$. The Gini of the uniform is equal to $\frac{(b-a)}{3(b+a)}$. We distinguish three sub-cases of how the distribution ``scales'' with $\bar{d}_U$; in all cases let us hold $\theta=1$.\\
(i) Suppose $b = 2\bar{d}_U-a$, and $a$ scales with $\bar{d}_U$, e.g. $a=\frac{\bar{d}_U}{2}$. It can be verified graphically that $f$ is increasing in $\bar{d}_U$. As we increase $\bar{d}_U$ with this parameterization, the CDF moves to the right and rotates downwards, hence it corresponds to a FOSD shift. The Gini coefficient is constant, equal to $1/6$, and hence unaffected by changes in $\bar{d}_U$.\\
(ii) Suppose that $b=\bar{d}_U+l/2$, $a=\bar{d}_U-l/2$, for some fixed $l \in (0, 2\bar{d}_U)$. It can be verified graphically that $f$ is increasing in $\bar{d}_U$. As we increase $\bar{d}_U$ with this parameterization, the CDF shifts (perfectly undistorted) to the right, hence this also corresponds to a FOSD shift. The Gini coefficient is equal to $\frac{l}{6\bar{d}_U}$, hence it is decreasing in $\bar{d}_U$.\\
(iii) Suppose finally that $b = 2\bar{d}_U-a$, as in case (i), but now $a$ stays constant. For $a=0$, we can show analytically that $f$ is increasing in $\bar{d}_U$; for `intermediate values of $a \in (1,10)$ it can be verified graphically that $f$ exhibits an inverted-U shape in $\bar{d}_U$; for `large' values of $a \geq 12$, it can be verified graphically that $f$ is decreasing in $\bar{d}_U$. As we increase $\bar{d}_U$ with this parameterization, the CDF rotates downwards, hence it corresponds to a FOSD shift. The Gini coefficient is equal to $\frac{1}{3}\left(1-\frac{a}{\bar{d}_U}\right)$, hence it is increasing in $\bar{d}_U$.

\section{Aggregating vacancies to ``locations''}\label{sec:app-E}
\begin{proposition}In a large market with locations, when applicant search intensities $\{\bar{d}_i\}_{i=1}^\infty$ are drawn i.i.d. from a distribution $G$, with a finite mean $\bar{d}_U$, support in a subset of $[0, \infty)$, and a moment-generating function (mgf) well-defined for non-positive values, and vacancies at locations $\{v_j\}_{j=1}^\infty$ are drawn i.i.d. from a distribution $H$ with support in $\mathbb{N}$ and mean $\bar{v}$, the mean job-finding probability is given by
\[
f = 1- M_{\bar{d}}(-\chi), 
\]
where $M_{\bar{d}}(t) \equiv \mathbb{E}[e^{\bar{d} t}]$ is the mgf of $\bar{d}_i, \forall i$; $\chi \equiv \mathbb{E}_{(v,\tilde{d})}\left[\min\left\{1, \frac{v}{1+\tilde{d}}\right\}\right]$, $\tilde{d} \sim Poisson(\bar{v}\cdot\bar{d}_U/\theta)$.
\end{proposition}
\begin{proof}As with Theorem 1, we start from the small market and find the limit of the large market. With locations, the job-finding probability of applicant $i$ in the small market is
\[
f_i = 1-\prod_{j=1}^L(1-p_i\phi_{j}(p_{-i}))
\]
where $\phi_{j}(p_{-i}) = \mathbb{E}_{\tilde{d}_i}\left[\min\left\{1, \frac{v_j}{1+\tilde{d}_i}\right\}\right]$, and $\tilde{d}_i\sim PB(\vec{p}_{-i})$.

Rearrange and take logs:
\[
\log(1-f_i) = \sum_{j=1}^L \log(1-p_i\phi_j(p_{-i}))
\]
Now, Taylor's theorem says
\[
\log(1-x) = -x -h(x)x
\]
where $h(x)x$ is (the Peano form of) the remainder, for which $\lim_{x\to 0} h(x) = 0$.

We can thus re-write the above as
\begin{align*}
\log(1-f_i) &= -p_i\sum_{j=1}^L \phi_j(p_{-i})\big(1+h_{j}\big)\\
			&= -\bar{d}_i\sum_{j=1}^L \frac{\phi_j(p_{-i})}{L}\big(1+h_{j}\big)
\end{align*}
where $h_{j} \equiv h\big(p_i\phi_j(p_{-i})\big)$.

Notice that since we keep the remainder in the expression, the above is an equality, not an approximation. Now, to go to the large market (Definition 4), we let $p_i \to 0$ and $L \to \infty$; thus the remainders $h_{j}$ go to $0$, so they drop from the sum. Also, we know that the dependence of $\phi_j$ on $p_{-i}$ drops, since $\tilde{d}_i \sim Poisson(\bar{d}_L), \forall i$, where $\bar{d}_L$ is the mean indirect degree at each location, the analog of $\bar{d}_V$ in the analysis at the vacancy level. We thus get that in the limit, 
\[
f_i = 1-e^{-\bar{d}_i\chi}
\]
where $\chi \equiv \lim_{L\to \infty} \sum_{j=1}^L \frac{\phi_j}{L}$, assuming this limit exists. $\phi_j$ are given by
\[
\phi_j \equiv \mathbb{E}_{\tilde{d}}\left[\min\left\{1, \frac{v_j}{1+\tilde{d}}\right\}\right]
\]
where $\tilde{d} \sim Poisson(\bar{d}_L)$; analogously to \Cref{lemma-acc-id}, consistency in degrees in the bipartite network implies that $\bar{d}_L=\bar{d}_U\cdot U/L$.

Now, assuming $v_j$ are drawn from a distribution $H$ with mean $\bar{v}$, we have  
\[
\chi = \mathbb{E}_{(v, \tilde{d})}\left[\min\left\{1, \frac{v}{1+\tilde{d}}\right\}\right]
\]
And the normalized number of unemployed is given by
\begin{align*}
U/L &= U/V\cdot V/L\\
	&= \theta^{-1}\cdot \bar{v}
\end{align*}
When search intensities $\bar{d}_i$ are drawn from a distribution $G$ with mean $\bar{d}_U$, the mean job-finding probability $f = \mathbb{E}[f_i]$ becomes 
\[
f = 1-M_{\bar{d}}(-\chi)
\]
where $\chi = \mathbb{E}_{(v, \tilde{d})}\left[\min\left\{1, \frac{v}{1+\tilde{d}}\right\}\right]$.
\end{proof}

\begin{proposition}Suppose $G'$ is a mean-preserving spread of $G$, where $G, G'$ are applicant search intensity distributions satisfying the conditions of Proposition 7. Then, for the corresponding mean job-finding probabilities in the large market with locations, it holds that
\[
f' < f
\]
\end{proposition}
\begin{proof}A MPS in $\bar{d}$ does not affect the mean $\bar{d}_U$, and hence it does not affect the distribution of indirect degrees $\tilde{d}$. Given that by assumption $H$ is not changed, it follows that $\chi$ stays the same. The result follows entirely analogously to the proof of Theorem 2, coming from the function $1-e^{-\bar{d}\chi}$ being increasing and concave in $\bar{d}$.
\end{proof}

\begin{proposition}Suppose $H'$ is a mean-preserving spread of $H$, where $H, H'$ are vacancy number distributions satisfying the conditions of Proposition 7. Then, for the corresponding mean job-finding probabilities in the large market with locations, it holds that
\[
f' \leq f
\]
\end{proposition}
\begin{proof} We have
\begin{align*}
\chi &=	\mathbb{E}_{\tilde{d}}\left[\mathbb{E}_{v}\left[\min\left\{1, \frac{v}{1+\tilde{d}}\right\} \bigg| \tilde{d}\ \right]\right]
\end{align*}
Conditional on $\tilde{d}$, the function $\min\left\{1, \frac{v}{1+\tilde{d}}\right\}$ is a weakly concave function of $v$: it is a ray from $(0,0)$ to $(1+\tilde{d},1)$ and then stays constant at $1$ for any $v \geq 1+\tilde{d}$. Denoting $v' \sim H'$ and $v \sim H$, it follows that
\begin{align*}
\mathbb{E}_{v'}\left[\min\left\{1, \frac{v'}{1+\tilde{d}}\right\} \bigg| \tilde{d}\ \right] &\leq \mathbb{E}_{v}\left[\min\left\{1, \frac{v}{1+\tilde{d}}\right\} \bigg| \tilde{d}\ \right] \Rightarrow\\
\mathbb{E}_{\tilde{d}}\left[\mathbb{E}_{v'}\left[\min\left\{1, \frac{v'}{1+\tilde{d}}\right\} \bigg| \tilde{d}\ \right]\right] &\leq \mathbb{E}_{\tilde{d}}\left[\mathbb{E}_{v}\left[\min\left\{1, \frac{v}{1+\tilde{d}}\right\} \bigg| \tilde{d}\ \right]\right]
\end{align*}
Denoting the LHS by $\chi'$ and the RHS by $\chi$, we have $\chi' \leq \chi$, and hence
\[
f' \leq f
\]
since $f = 1-M_{\bar{d}}(-\chi)$ is an increasing function in $\chi$.
\end{proof}

\section{Matching with heterogeneous job advertising}\label{sec:app-F}
\begin{proposition}In a large market, when job-advertising intensities $\{\bar{d}_j\}_{j=1}^\infty$ are drawn i.i.d. from a distribution $\hat{G}$, with a finite mean $\bar{d}_V$, support in a subset of $[0, \infty)$, and a moment-generating function (mgf) well-defined for non-positive values, the mean vacancy-filling probability is given by
\[
q = \frac{1-e^{-\theta(1-M_{\bar{d}}(-1))}}{\theta}, \tag{2}
\]
where $M_{\bar{d}}(t) \equiv \mathbb{E}[e^{\bar{d} t}]$ is the mgf of $\bar{d}_j, \ \forall j$.
\end{proposition}
\begin{proof}As with Theorem 1 and Proposition 7, we start with the small market and take the limit. In the case of search heterogeneity on the vacancy side, all applicants are symmetric. An applicant's degree distribution is $PB(\vec{p})$ with mean $\bar{d}_U \equiv \sum_j p_j$. Now, as long as the vacancy has at least one link, it makes an offer for sure; in fact it will make just this one offer. Hence the probability that a vacancy gets filled is given by
\[
q_j = \big(1-(1-p_j)^U\big)\psi(p_{-j})
\]
where $\psi(p_{-j})$ is the probability the applicant who receives the offer accepts it. Since applicants are all symmetric, $\psi(p_{-j})$ does not depend on which applicant vacancy $j$ is offered to, though it depends on the search intensities of the other vacancies. It holds that
\[
\psi(p_{-j}) = \mathbb{E}\left[\frac{1}{1+\tilde{o}}\right]
\]
where $\tilde{o}$ is the (random) number of \textit{extra} offers received by the candidate to whom vacancy $j$ is offered.

The number of offers $o$ an applicant receives follows $PB(\vec{z})$, and the number of \textit{extra} offers $\tilde{o}$ received by someone vacancy $j$ connects to, follows a $PB(\vec{z}_{-j})$ distribution, where
\begin{align*}
z_k &\equiv Pr\{\text{receive offer from vacancy k}\} = p_k\mathbb{E}\left[\frac{1}{1+\tilde{d}_k}\right]
\end{align*}
and $\tilde{d}_k \sim Bin(U-1, p_k), \ k=1..V$. We know that $\mathbb{E}\left[\frac{1}{1+\tilde{d}_k}\right] = \frac{1-(1-p_k)^U}{p_kU}$. 

This concludes the small market treatment when there is heterogeneity on firm recruitment intensity. 

Notice that $\mathbb{E}\left[\frac{1}{1+\tilde{d}_j}\right]$ is the $Pr\{\text{receive an offer from j} \ | \ \text{linking to j}\}$. This probability is the one we had been denoting by $\phi$ and was constant across firms in our previous analysis (though it varied by applicant), now varies by firm. In the large market, defined entirely analogously to Definition 3, we know that $\psi(p_{-j}) \to \psi$, where
\[
\psi = \mathbb{E}\left[\frac{1}{1+\tilde{o}}\right]
\]
and $\tilde{o} \sim Poisson(\sum_j z_j)$. We still need to compute $\sum_j z_j$. It holds that
\begin{align*}
\sum_j z_j &= \sum_j \frac{1-\left(1-\frac{\bar{d}_j}{U}\right)^U}{U}\\
			&= \frac{V}{U}\sum_j \frac{1-\left(1-\frac{\bar{d}_j}{U}\right)^U}{V}\\
			&\to \theta \sum_j \frac{1-e^{-\bar{d}_j}}{V} 
\end{align*}
where the last line is the limit when $U,V\to \infty$. Analogously to the derivation for $\phi$ in the proof of Lemma 5, it follows that
\[
\psi = \frac{1-e^{-\sum_j z_j}}{\sum_j z_j}
\]

We are done. The vacancy-filling probability in the large market, conditional on a level $\bar{d}_j$ of vacancy-search/advertising intensity is
\[
q_j = (1-e^{-\bar{d}_j})\psi
\]
Assuming $\bar{d}_j$'s are drawn from a distribution $\hat{G}$, we get from the Law of Large numbers that $\sum_j \frac{1-e^{-\bar{d}_j}}{V} \to \mathbb{E}[1-e^{-\bar{d}_j}]$, yielding
\begin{align*}
q &= \left(1 - M_{\bar{d}_j}(-1)\right)\frac{1-e^{-\theta(1-M_{\bar{d}_j}(-1))}}{\theta (1-M_{\bar{d}_j}(-1))}\\
  &= \frac{1-e^{-\theta(1-M_{\bar{d}_j}(-1))}}{\theta}
\end{align*}
We drop the $j$ for notational convenience.
\end{proof}

\begin{corollary}When $\hat{G}$ is degenerate at $\bar{d}_V$ and $G$ is degenerate at $\bar{d}_U=\theta\bar{d}_V$, the job-finding probability of (1) and the vacancy-filling probability of (2) satisfy
\[
f=\theta q
\]
\end{corollary}
\begin{proof} When $\hat{G}$ is degenerate, formula (2) of Proposition 10 becomes $q = \frac{1-e^{-\theta(1-e^{-\bar{d}_V})}}{\theta}$; using the fact that $\bar{d}_V=\bar{d}_U/\theta$, and that $\phi=\frac{1-e^{-\bar{d}_V}}{\bar{d}_V}$ this simplifies to 
\begin{align*}
q &= \frac{1-e^{-\theta(1-e^{-\bar{d}_V})}}{\theta}\\
  &= \frac{1-e^{-\bar{d}_U\frac{1-e^{-\bar{d}_V}}{\bar{d}_V}}}{\theta}\\
  &= \frac{1-e^{-\bar{d}_U\phi}}{\theta}
\end{align*}
The numerator is our expression for $f$, when $\bar{d}_i$ are drawn from the degenerate distribution (see the first special case in \Cref{subsec:spec-cases}).
\end{proof}

\begin{proposition}Suppose $\hat{G}'$ is a mean-preserving spread of $\hat{G}$, where $\hat{G}$ and $\hat{G}'$ are job-advertising intensity distributions satisfying the conditions of Proposition 10. Then, for the corresponding mean vacancy-filling probabilities in the large market, it holds that
\[
q' < q
\]
\end{proposition}
\begin{proof}The proof is entirely analogous to that of Theorem 2: the function $1-e^{-\bar{d}}$ is increasing and concave in $\bar{d}$, thus a MPS translates in a lower $q$.
\end{proof}
\end{appendices}

\bibliographystyle{aea}
\bibliography{NetSMBib}

\end{document}